\documentclass[a4paper,10pt,twoside]{article}
\usepackage[top=3cm, bottom=3cm, inner=3cm, outer=3cm, includehead]{geometry}
\usepackage{fancyhdr}
\usepackage{xurl}
\usepackage{graphicx}
\usepackage{alltt}
\usepackage{amsmath}
\usepackage[hidelinks, pdftex]{hyperref}
\usepackage[T1]{fontenc}
\usepackage[utf8]{inputenc}
\usepackage{lmodern}
\usepackage{IEEEtrantools}
\usepackage{csquotes}
\usepackage[backend=biber, style=numeric, sorting=none]{biblatex}
\usepackage{lipsum}
\usepackage[utf8]{inputenc}
\usepackage{xcolor}
\usepackage{lipsum}
\usepackage[T1]{fontenc}
\usepackage[utf8]{inputenc}
\usepackage{amsmath, amssymb, amsthm}
\usepackage[expansion=false]{microtype}
\usepackage{xcolor}
\usepackage{booktabs}
\usepackage{array}
\usepackage{tcolorbox}
\usepackage{enumitem}
\usepackage{hyperref}
\usepackage{titlesec}
\usepackage{fancyhdr}
\usepackage{parskip}
\usepackage{amsthm}
\usepackage{amssymb}
\usepackage{mathtools}
\tcbuselibrary{theorems, skins, breakable}
\newtheorem{theorem}{Theorem}[section]
\newtheorem{Proposition}[theorem]{Proposition}
\newtheorem{corollary}{Corollary}[theorem]
\newtheorem{lemma}[theorem]{Lemma}
\newtheorem{assumption}[theorem]{Assumption}
\newtheorem{definition}{Definition}[section]
\newtheorem{remark}{Remark}
\newcommand{\R}{\mathbb{R}}
\newcommand{\RP}{\mathbb{RP}}
\newcommand{\T}{\mathrm{T}}
\newcommand{\Cayley}{\mathcal{C}}
\newcommand{\diag}{\mathrm{diag}}
\newcommand{\norm}[1]{\left\|#1\right\|}

\newcommand{\sqmu}{\sqrt{\mu}}
\newcommand{\sqmuL}{\sqrt{\mu/L}}
\newcommand{\sqLmu}{\sqrt{L/\mu}}
\newcommand{\sqkappa}{\sqrt{\kappa}}
\newcommand{\xst}{x^*}
\newcommand{\nbf}{\nabla^2\! f}

\newcommand{\phist}{\phi_k^*}
\newcommand{\lf}{\ell_f}
\newcommand{\eps}{\varepsilon}
\newcommand{\Mc}{\mathcal{M}}
\newcommand{\Dop}{\mathcal{D}}
\newcommand{\Gr}{\mathrm{Gr}}
\DeclareMathOperator{\Sch}{S}

\begin{document}

\title{On the Nesterov's acceleration: A NAIM perspective}

\author{Rachit Mehra\thanks{Project Lead, TenneT Offshore GmBH, Germany}  
\and 
M Parimi \thanks{Research Scholar, $E-MC^2$ Lab,VJTI, India} 
\and
Amol Yerudkar \thanks{Associate Professor, School of Computer Science and Technology, Zhejiang Normal University, China}
\and
S.R. Wagh \thanks{Faculty- EED \& P.I., $E-MC^2$ Lab, VJTI, India }
\and
Navdeep Singh \thanks{IGI Research Chair Professor, $E-MC^2$ Lab, VJTI, India }
}

\maketitle

\begin{abstract} 
We present a unifying Nearly Asymptotically Invariant Manifold (NAIM) framework for
understanding Nesterov’s Accelerated Gradient (NAG) method. By lifting the first-order gradient flow into a second-order phase space we construct a NAIM--a slow, attracting graph and show that acceleration emerges from a curvature aware perturbation of this graph. The evolving slope of the perturbed manifold is governed by a Differential Riccati Equation (DRE), which enforces strict tangency of the vector field to the manifold surface. In the quadratic case the DRE reduces to an Algebraic Riccati Equation (ARE), and the requirement of spectral resonance—equal contraction rates across all curvature modes—uniquely determines the damping coefficient, directly
yielding the continuous-time Nesterov ODE. Fenichel’s theorem then extends this picture
rigorously to general smooth, strongly convex landscapes: normal hyperbolicity guarantees
persistence of the accelerated manifold despite varying Hessian curvature. The method is further extended to unified geometric derivation of NAG methods for smooth convex and strongly convex optimization in the discrete case. We exploit the underlying geometric structure and derive both cases from the same principle of preserving the projective structure under discretization process. A Lie--Trotter splitting separates the linear dissipative dynamics from the nonlinear gradient flow. The dissipative subsystem is integrated by the Cayley (bilinear)
transform, which preserves the underlying projective (M\"{o}bius) structure unconditionally and
produces the classical Nesterov momentum coefficient as the unique Pad\'{e} multiplier. For the convex case, projective flatness (vanishing Schwarzian derivative) uniquely selects the time-varying damping recovering the canonical Nesterov ODE for convex functions.
\end{abstract}
\paragraph{Keywords:}
Nesterov's accelerated scheme, convex optimization, Normally Attractive Invariant Manifold, Fenichel's theorem

\tableofcontents

\section{Introduction}

Standard first-order optimization methods  struggle with ill-conditioned, ravine-like landscapes. 
While NAG method overcomes this topography, its traditional mathematical justification, the estimation sequence, is epistemically incomplete. It relies on an unexplained, assumed quadratic structure and discovers the optimal momentum merely as an algebraic fixed point, failing to provide any physical or geometric mechanism for why the acceleration actually works. \par

Despite the algorithm’s maturity, the geometric mechanism underlying the acceleration has
remained elusive. \cite{su2016differential} linked Nesterov’s method to a second-order ODE
but the discretization and extension to non-quadratic $f$ remained informal. \cite{Shi2019} derived high-resolution ODEs capturing momentum differences. Despite the method’s maturity, some gaps in its justification persist. 
\begin{enumerate}
    \item \textit{Logical Gap in Acceleration Proofs:} The classical estimation-sequence proof \cite{nesterov2013introductory}, \cite{nesterov2018lectures} imposes a quadratic lower-bound ansatz without geometric justification and identifies the optimal momentum only as an algebraic fixed point (see Appendix \ref{App:Gap} for full discussion). There are obstacles in generalizing to Non-Quadratic functions. While acceleration admits an exact and transparent description for quadratic objectives, extending the analysis to general strongly convex functions introduces substantial mathematical complications. The Hessian matrix becomes time-varying along the optimization trajectory, implying that the spectral structure governing optimal decay is only locally valid. Consequently, the Riccati equation describing the dynamics becomes a variable-coefficient partial differential equation without a closed-form solution, and the invariant structure guiding the dynamics is no longer a linear subspace but a curved nonlinear manifold.
    \item \textit{Spectral mechanism:} The damping value is asserted in ODE-based derivations \cite{su2016differential}, \cite{Shi2019} but never derived from a structural principle.
    \item \textit{Limitations of Physical Analogies:} Simple mechanical interpretations, such as the Heavy-Ball analogy \cite{polyak1964some}, \cite{Arun} ultimately fail to capture the mechanism of acceleration. These models lack the ability to anticipate curvature changes in the objective landscape, which leads to overshooting near the minimum and generates instability and oscillatory behavior. Moreover both NAG and Heavy Ball algorithms arise from the same ODE yet Heavy–Ball diverges for non-quadratic function \cite{lessard2016analysis} while Nesterov converges. The geometric reason has not been explained at the discrete level.
    \item \textit{Discretization without geometric justification:} Existing discretizations of the Nesterov ODE use Euler methods \cite{Shi2019, JMLR_Su} or variational integrators \cite{wibisono2016variational} but are without geometric justification. Conventional numerical integrators, particularly the explicit Euler method, are structurally inadequate for this purpose. Euler discretization uses the gradient only as a local downhill direction and neglects the geometry of the constraint surface. As a consequence, it produces structural defects such as energy drift, departure from invariant manifolds, and violations of geometric invariants. 
\end{enumerate}


By lifting optimization into a second-order phase space with velocity, this work constructs a NAIM framework. The framework enhances standard optimization by introducing a specific, curvature-aware perturbation to the unaccelerated invariant manifold, creating a steeper geometric track that drives faster convergence.  To maintain stable convergence along this new geometric route, strict tangency between the dynamic vector field and the manifold's surface needs to be enforced. The graph's evolving slope is governed by the DRE, which acts as the strict tilt consistency condition that enforces this tangency. The perturbed system must be appropriately damped to ensure stability and adherence to the manifold; the required damping coefficient is not heuristic but emerges as the stable root of the Riccati equation, corresponding to a spectral resonance phenomenon that uniformly equalizes contraction rates across all directions. Applying this geometrically derived damping yields the final continuous-time ordinary differential equation (ODE) of NAG method. To the best of authors knowledge this provides the first purely geometric, structural proof of the so called acceleration step. The proposed framework could be logically extended  to general non-quadratic landscapes through Fenichel's theorem. Finally, to prevent the energy drift and step-size restrictions caused by geometrically blind numerical methods, the framework employs structure-preserving discretization. This involves modularly separating the linear momentum dynamics from the non-linear gradient flow—to ensure that the theoretical speedup translates into the algorithm. 

The strongly convex case is discretized using structure-preserving Lie-Trotter splitting \cite{Blanes2024} and the Cayley transform \cite{VIRGOS2018} to unconditionally stabilize momentum, whereas the convex case is discretized by enforcing projective flatness to yield the geometrically canonical Nesterov ODE. This work justifies the final ODE as a form amenable to structure-preserving discretizations, to ensure the theoretical speedup translates to the actual algorithm without drifting off course.

The present work critiques the mechanistic viewpoint of acceleration, arguing that it obscures the underlying geometric intuition. To address the gaps, we make the following contributions:
\begin{enumerate}
    \item \textit{Formulating acceleration as NAIM graph evolution (Subsection: \ref{subsec:Grpah}):} The work successfully reconceptualizes acceleration by lifting the first-order gradient flow into a second-order phase space, framing the optimal optimization path as a NAIM. The Riccati equation is derived as a strict \textit{tilt consistency condition} that calculates the exact slope required to keep the algorithm perfectly tangent to the optimal path.
    \item  \textit{Spectral resonance arriving at the continuous-time ODE (Lemma: \ref{OptDam}):} Instead of discovering the optimal momentum merely as an algebraic coincidence, the framework evaluates the evolution of the curvature aware perturbed graph to find the exact slope needed for critical damping. This geometrically derived optimal damping triggers spectral resonance, equalizing contraction rates across all dimensions. Applying this specific resonant damping  directly arrives at the continuous-time ODE of Nesterov's method. 
    \item \textit{Closing the logical gap of the estimation sequence (Appendix: \ref{App:Gap}):} The traditional estimation sequence proof is shown to be mathematically valid but epistemically incomplete because it relies on an unexplained quadratic ansatz and discovers the optimal momentum merely as an algebraic coincidence. This work replaces that algebraic inequality transfer with a complete physical and geometric mechanism for why acceleration works.
    \item \textit{Employing Fenichel’s theorem for general landscapes (Theorem \ref{thm:fenichel}):} Instead of relying on a quadratic ansatz to bridge the gap between simple quadratic problems and general non-quadratic landscapes, this work uses Fenichel's theorem, a key result in Geometric singular perturbation theory. Because the algorithm's fast transverse attraction to the manifold strictly dominates any tangential drift (normal hyperbolicity), Fenichel's theorem rigorously guarantees that the optimal accelerated path persists despite varying curvatures. The optimal convergence rate is shown to emerge physically from this continuous$-$time Fenichel spectral gap.
    \item \textit{Structure-preserving discretization (Section \ref{sec:SPD}):} The work demonstrates that standard numerical discretization (like explicit Forward Euler) is geometrically blind to the invariant manifold constraint, which causes energy drift and step-size limitations that cancel out the theoretical acceleration. As a solution, the work contributes a Lie-Trotter splitting method combined with the unconditionally stable Cayley transform to discretize the continuous geometry, perfectly translating the speedup into the actual algorithm without drifting off course.
    \item  \textit{Resolving the strictly convex case via projective geometry (Subsection \ref{subsec:sconvex}):} Finally, for smooth convex functions where strong convexity is absent, the work introduces a novel perspective grounded in projective geometry. It shows that the exact, time-varying damping required for standard convex acceleration emerges naturally by enforcing projective flatness (a zero Schwarzian derivative), which uniquely selects the canonical Nesterov constants.
    \item \textit{Fenichel verification (Proposition \ref{prop:hypotheses}) :} We verify all five hypotheses of Fenichel’s theorem for the NAG fast-slow system on compact sublevel sets.
    \item \textit{Triple-momentum rate ceiling (Corollary \ref{TM}) :}. A cubic Riccati factorization theorem proves that the triple momentum methods \cite{Scoy2018} cannot exceed the NAG method convergence rate.
    \end{enumerate}

The manuscript is organized as follows: In Section \ref{sec:GeoFrame} the NAIM framework is proposed for continuous case and Nesterov's ODE is rigorously derived. In Section \ref{sec:fenichel} a generalization to composite optimization, where optimization function may be non-smooth, is obtained. The discretization of the continuous time ODE for the strongly convex and the convex case is further derived in Section \ref{sec:SPD}. Sections 2–4 develop the core results in self-contained form. All detailed proofs, supporting analysis, and supplementary material are collected in the appendices. Appendices include the NAIM definition on fiber bundles (appendix \ref{App:Fiber}), the Riccati tilt derivation (appendix \ref{App:Tilt}),  resolution of gaps in NAG (appendix \ref{App:Gap}), verification of Fenichel theorem (appendix \ref{App:fenichel}), note on triple momentum dynamics methods (appendix \ref{App:Triple}) and analysis on perturbed NAIMs (appendix \ref{App:Unperturbed}).

\section{NAIM geometric framework}\label{sec:GeoFrame}

Consider the unconstrained optimization problem
\[
\min_{x \in \mathbb{R}^n} f(x)
\]
where $f: \R^n \to R$ is $L$-smooth and $\mu$ strongly convex with $0 < \mu \leq L$. The condition number $\kappa = L /\mu$ characterizes the difficulty of the problem. The continuous time gradient flow $x= - \nabla f(x)$ converges to the minimizer but suffers from slow convergence for ill conditioned problems. Nesterov's accelerated gradient method achieves the optimal convergence rate, and its continuous-time limit takes the form of a damped oscillator equation. 

For a function $f(x)$, we have:
$$ \mu I \le \nabla^2 f(x) \le L I $$

and the corresponding $1/L$-scaled gradient flow of $f(x)$ is defined as:
\begin{equation} \label{eq:flow}
    \dot{x} = -\frac{1}{L} \nabla f(x)
\end{equation}

The second-order ODE takes the following form:

\begin{equation} \label{eq:Nest}
    \ddot{x} + 2 \sqrt{\frac{\mu}{L} } \dot{x} + \nabla f(x) =0    
\end{equation}

The objective of the section is to reinterpret \eqref{eq:Nest} geometrically and to show how acceleration emerges from a structured perturbation of its lifted realization. We show that \eqref{eq:Nest} can be understood geometrically as the dynamics induced by a tilted attracting slow manifold in a lifted state space. The construction proceeds as follows:
\begin{enumerate}
    \item Start from the $1/L$-scaled gradient flow and lift the dynamics to a second-order state space. Realize gradient flow as a slow invariant graph.
    \item Introduce adapted coordinates separating tangential and normal directions. Perturb the normal dynamics so that the invariant graph tilts and show that this tilt induces an accelerated effective descent scale. Derive a Riccati-type invariance condition that balances graph slope, curvature, and normal contraction.
    \item In the quadratic case, recover an exact ARE and obtain the square-root damping law by critical damping of the weakest curvature mode.
\end{enumerate}

Concluding we show that the acceleration is not merely an inertial modification of gradient descent, but rather the manifestation of a curvature-aligned invariant manifold in a lifted dynamical system

\subsection{Step 1: Lifting to Second Order Structures}
The gradient descent algorithm which minimizes $f(x)$ is constrained by the smoothness $L$. Hence, across ill-conditioned level sets, the iterate update is slow along flat terrains and fast along steeper terrains. The issues with getting stuck at saddle nodes remain. Hence the need for algorithms with structures which can be exploited during iterate updates and this is possible only through lifted coordinates. Though this doesn't change the dimension of $f(x)$, it adds dimension to the algorithm which finds the minimum of $f(x)$.

Introduce lifted coordinates $(x_1, x_2) \in \mathbb{R}^n \times \mathbb{R}^n$ and consider the second-order system:
\begin{align}\label{eq:1}
    \dot{x}_1 &= x_2 \\
    \dot{x}_2 &= u_{in} + u_{a}
\end{align}
where $u_{in}$ and $u_{a}$ are virtual controls. 
The lifted coordinates should finally result in convergence to the manifold $$ M_0 = \{ (x_1, x_2) \in R^{2}\mid x_2 + \frac{1}{L} \nabla f(x_1) = 0 \}$$
on which $\dot x_1=-\frac{1}{L} \nabla f(x_1) $ recovering \ref{eq:flow}. 

The objective of $u_{in}$ and $u_a$ is to design them in such a way as to ensure that the desired manifold $M_0$ is invariant and attractive.
\begin{enumerate}
    \item The invariant property ensures that once a trajectory is on the manifold $M_0$, it stays on $M_0$. The virtual control $u_{in}$ is designed to make the manifold $M_0$ invariant.

    We define $z$ as the transverse coordinate which measures how far the state $(x_1,x_2)$ is from the candidate manifold $M_0$, which is defined by ${z=0}$.
    
        \begin{align} \label{eq:3}
        \dot{x}_1 &= z-\frac{1}{L} \nabla f(x_1)  \nonumber\\ 
        \dot{z} &= - \alpha z
        \end{align}
    
    On the manifold $M_0$ the invariance condition implies $z=0 \implies \dot{z}=0$.
    $$ \dot{z} = \dot{x}_2 + \frac{1}{L} \nabla^2 f(x_1) \dot{x}_1 = u_{in} + \frac{1}{L} \nabla^2 f(x_1) x_2 $$
    Setting $z=0$ gives the design for the internal control $u_{in}$:
    \begin{equation}
     u_{in} = -\frac{1}{L} \nabla^2 f(x_1) x_2 \label{uin}   
    \end{equation}
    As a control law $u_{in}$ cancels exactly the drift of $z$ along the system flow, keeping trajectories on $M_0$ once they reach it. 
    \item Trajectories that are not on the manifold are naturally and transversally contracted back towards it. To ensure attractivity of the manifold, we set the dynamics of $z$ to be:
    $$ \dot{z} = -\mu z $$ which makes $z(t)=z(0) e^{-\mu t} \to 0$ exponentially (where $\mu$ is the rate of contraction). 
    Substituting back:
    \begin{equation} \label{ua}
        u_{a} = -\mu \left( x_2 + \frac{1}{L} \nabla f(x_1) \right) = -\mu z
    \end{equation}
   The transverse derivative of $z$ is strictly negative, pulling all off-manifold trajectories back toward $M_0$.
\end{enumerate}

The full virtual control is the orthogonal sum of the two components:
\begin{equation}\label{eq:u_total}
  u \;=\; u_{im} + u_{a}
    \;=\; -\frac{1}{L}\,\nabla ^2 f(x_{1})x_{2}
          \;-\;\mu\!\left(x_{2}+\frac{1}{L}\,\nabla f(x_{1})\right).
\end{equation}
where $u_{im}$ acts tangentially to $M_0$ (maintains structure) and $u_{a}$ acts transversely to $M_0$ (drives convergence onto $M_0$).

\begin{remark}[Time-scale separation] The unperturbed manifold $M_0$ represents the slow dynamics of the system. This evolves at a slow time scale related to the inverse of the smoothness constant $L$ (i.e., a rate of $O(1/L)$).  The transverse fiber dynamics governed by \eqref{ua} contracts at rate $\mu$. The ratio $\kappa = L/ \mu$ quantifies the time-scale separation that underlies the NAIM structure.
\end{remark}

Replacing $z$ with $v$ to show representation in adapted coordinates as:
\begin{align} \label{eq:4}
\dot{x}_1 &= v-\frac{1}{L} \nabla f(x_1)  
\end{align}
\begin{align}\label{eq:5}
\dot{v} &= - \alpha v
\end{align} 
where $\alpha > 0 $ is the transverse contraction rate. The system~\eqref{eq:4}--\eqref{eq:5} is the NAIM with slow dynamics along $M_0$ and fast fibers contracting exponentially.


\subsection{Step 2: Representing NAIM as a Graph} \label{subsec:Grpah}

In the lifted space $(x, v) \in \R^n \times \R^n$, the state naturally separates into a base
component $x$ (position, governing tangential progress) and a fiber component $v$ (velocity,
governing transverse attraction).
 
\begin{definition}[Graph manifold]
An invariant graph manifold with slope operator $P \in \R^{n \times n}$ is
\[
  \mathcal{M}_P = \{(x,v) \in \R^{2n} \mid v = Px\}.
\]
The slope $P = D\sigma(x)$ is the Jacobian of the graph map $\sigma : \R^n \to \R^n$.
\end{definition}
 
\noindent For the unperturbed manifold $\mathcal{M}_0$, the slope is $P_0 = -\tfrac{1}{L}\nabla^2 f(x_1)$, which is negative semidefinite by $L$-smoothness.
 
\begin{remark}[Grassmannian interpretation]
Tracking the slope $P(t)$ is equivalent to tracing a path on the Grassmannian
$\mathrm{Gr}(n, 2n)$---the space of all $n$-dimensional subspaces of $\R^{2n}$. The DRE
below is the equation of a geodesic flow on $\mathrm{Gr}(n, 2n)$ with the Riemannian structure
induced by the Hessian.
\end{remark}
 
\begin{remark}[Non-quadratic landscapes]
For a general smooth, strongly convex $f$, the Hessian $H_f(x_1(t))$ varies along the
trajectory and the invariant manifold cannot maintain a constant slope. The graph must be
represented by a time-varying relation $v = P(t)x$, and the tangency condition expands into
the DRE introduced in Section~\ref{subsec:SR}.
\end{remark}
 
The unperturbed system~\eqref{eq:4}--\eqref{eq:5} converges to $\mathcal{M}_0$ but inherits the same slow tangential rate $O(1/L)$ as gradient descent. To achieve acceleration, we introduce the perturbation
\begin{equation}
  \dot{x}_2 = -\alpha x_2 - \nabla f(x_1), \label{eq:perturb}
\end{equation}
which modifies the fiber dynamics. In adapted coordinates, the perturbed system (retaining
leading-order terms in the fast--slow splitting) becomes:
\begin{align}
  \dot{x} &= v, \label{eq:ps1} \\
  \dot{v} &= -\alpha v - \nabla f(x). \label{eq:ps2}
\end{align}
\noindent To overcome this limitation and achieve acceleration, a curvature-aware perturbation is introduced, shifting the system onto a new perturbed manifold ($M_\varepsilon$) with a steeper slope that drives trajectories towards the optimum at a higher velocity. 

\begin{equation} \label{eq:8}
\dot{x}_2 = -\alpha x_2 - \nabla f(x_1) \nonumber
\end{equation}
This is equivalent to stating that, for acceleration:
\begin{align}\label{eq:9}
    \dot{x}_1 &= x_2 \nonumber\\
    \dot{x}_2 &= u_{in} + u_{a}+u_{per} \nonumber
\end{align}
where $u_{per}=-\alpha x_2 - \nabla f(x_1) $.

Let $M_\epsilon$ denote the perturbed invariant manifold. Representing the evaluation of the perturbation with respect to the adapted coordinates, from (\ref{eq:5}), 

\begin{equation*}
\dot{v} = -\alpha x_2 - \nabla f(x_1) + \frac{1}{L} \nabla^2 f(x_1) \dot{x}_1
\end{equation*}
\begin{align*}
\dot{v} &= -\alpha \left[ v - \frac{1}{L} \nabla f(x_1) \right] - \nabla f(x_1) + \frac{1}{L} \nabla^2 f(x_1) \left[ v - \frac{1}{L} \nabla f(x_1) \right] \\
&= -\alpha v - \nabla f(x_1) + \frac{1}{L} \nabla^2 f(x_1) v - \frac{1}{L^2} \nabla^2 f(x_1) \nabla f(x_1) + \frac{\alpha}{L} \nabla f(x_1)\\
&= \underbrace{-\alpha v - \nabla f(x_1)}_{\text{Leading order terms} }  +\underbrace{\frac{1}{L} \nabla^2 f(x_1) v - \frac{1}{L^2} \nabla^2 f(x_1) \nabla f(x_1) + \frac{\alpha}{L} \nabla f(x_1)}_{\text{ Higher order terms} }  
\end{align*}

The leading order refers to the first term in an asymptotic expansion in the small parameter $\epsilon$. For the slow manifold, this gives you the critical manifold $M_0$ (solution to the layer problem at $\epsilon=0$), which is the zeroth order approximation.

\begin{align}\label{eq:10}
    \dot{x}_1 &= v - \frac{1}{L} \nabla f(x_1)  
    \end{align}
    \begin{align}\label{eq:13}
    \dot{v} &= -\alpha v - \nabla f(x_1) + \text{Higher Order Terms} 
\end{align}

Since $v=x_2+\frac{1}{L} \nabla f(x_1)$, \\
\begin{equation}\label{eq:14}
    \dot v= \dot x_2+\frac{1}{L} \nabla f(x_1) \dot x_1 \\
    =-\alpha x_2 - \nabla f(x_1)+\frac{1}{L} \nabla f(x_1) \dot x_1 
\end{equation}

\noindent From (\ref{eq:13}): $\dot v=-\alpha v-\nabla f(x_1)$, indicating that when $\dot v=0, v=-\frac{1}{\alpha} \nabla f(x_1)$. Since $\dot x_2=-\alpha x_2 - \nabla f(x_1)$, $\dot x_2=0 \implies x_2=-\frac{1}{\alpha} \nabla f(x_1)$. This results in the conclusion $x_2=v$.\\
The final adapted coordinate representation of the perturbed dynamics are:
\begin{align}\label{eq:15a}
     \dot{x} &= v\\
     \dot v&=-\alpha v-\nabla f(x) \label{eq:15b}
\end{align}

\noindent The invariance of $M_\epsilon$ is now established. 

\begin{definition}[Adapted fiber coordinate]\label{def:w}
  \[
    w \;=\; v - Px
  \]
$w$ measures the distance from $M_0$ in the fiber direction: $w=0\Leftrightarrow(x,v)\in M_\epsilon$.
\end{definition}
\noindent The adapted fiber $w=v-Px$ measures the transverse deviation from the candidate manifold $M_\epsilon=\{(x,v) \mid v=Px\}$.\\

\noindent Conditions for invariance of  $M_\epsilon $ require: if $w=0\implies \dot w=0$.
\begin{align}\nonumber
    \dot w|_{w=0} &=\dot v-\dot P x-P \dot x \nonumber\\
    &=-\alpha(w+Px)-\nabla f(x)-\dot P x-P \dot x\nonumber\\
    &=-\alpha Px-H_f x-\dot P x-P^2 x (where \nabla f(x)\approx H_f x)\nonumber\\
    &=-(\alpha+P)w-[\dot P+P^2+\alpha P+H_f]x \nonumber
\end{align}
Thus, when $w=0$, $\dot w=0$ which requires:
\begin{equation}\label{eq:11}
    \dot P +P^2 +\alpha P+H_f =0
\end{equation}

Discussions:

\begin{enumerate}
    \item For the perturbed, steeper graph to be a valid track for the algorithm, it must be strictly invariant—meaning trajectories starting on it must never leave it. Geometrically, invariance requires that the optimization vector field remains tangent to the graph's surface everywhere. The rate of change in the fiber direction computed via the physical vector field must exactly match the rate of change computed by following the graph's slope along the base motion.
    \item It is required to check how the graph evolves under perturbation, ensuring that it remains orthogonal to $M_\epsilon$. This requires that the evolution of the graph due to perturbation should just grace through $M_\epsilon $ without destroying it. This behaviour can be evaluated through the evolution of the graph and ensuring that the evolution is critically damped. Since the graph has two coordinates $x$ and $v$ related through $P$, we evaluate the evolution of $P$ through the Riccati equation. The ensuing conditions for critically damped oscillations results in finding the slope and graph relation for $M_\epsilon$.    
    \item Equation (\ref{eq:11}) is a DRE, which indicates the evolution of the slope $P$. Instead of tracking the evolution of the base and fiber coordinates (in $R^{2n}$), one tracks the evolution of $n$-dimensional linear subspace that describes the relationship between $x$ and $v$. The space of all possible $n$-dimensional planes passing through the origin of a $2n$-dimensional space is a foundational concept in projective geometry called a Grassmannian manifold. $\dot{P}$ represents the literal velocity of the manifold's tilt as it traces a continuous path across the surface of the Grassmannian manifold $\Gr(n, 2n)$. The notion of projective geometry is crucial because it replaces the concept of absolute distance with the robust concept of relative ratio. Projective geometry works by focusing only on ratios and relative positions, ignoring absolute scale. This makes it ideal for flat or degenerate landscapes. 
    \item The addition of the perturbation $M_\epsilon$ has a profound effect on the system's dynamics. It creates a singularly perturbed system—one that evolves on two (or more) distinct time scales. The dynamics are split into two components. \begin{enumerate}
        \item The Fast Inner Layer Dynamics: These are the transverse dynamics, responsible for the \textit{Normally Attractive} property. They operate on a very fast time scale (related to the strong convexity constant $\mu$) and decay exponentially. Their sole purpose is to rapidly pull any off-manifold state $x$ onto the perturbed manifold $M_\epsilon$. 
        \item The Slow Perturbed Outer Layer: These are the dynamics on the manifold $M_\epsilon$ itself. This is the \textit{slow flow} that governs the system's evolution towards the optimal solution $x^*$.
    \end{enumerate}

    \item In the lifted Nesterov system, perturbative coupling between the slow coordinates $(x, v = Px)$ and the transverse fiber coordinate w may destroy uniform normal hyperbolicity of the nominal slow manifold. Nevertheless, the manifold persists in a weaker, nonuniform sense whenever the perturbation contributes sufficient dissipation in the transverse directions. In this regime, the transverse dynamics remain contracting on average, and one obtains an asymptotically attracting invariant graph $w = h(x)$, even though the classical uniform spectral gap condition fails. Thus the reduced slow dynamics retain geometric meaning not because the perturbation is small, but because it is geometrically stabilizing: transverse contraction compensates for time–scale mixing and ensures that trajectories are still funneled toward a nonuniform NAIM surrogate.
    \item The details are derived in the appendix: \ref{App:Tilt} where the Ricatti condition is derived from the tilt consistency condition.
\end{enumerate}

\subsection{Step 3: Spectral Resonance and the Nesterov ODE} \label{subsec:SR}
\begin{Proposition}
    (Riccati Invariance Condition). For the manifold $\{w = 0\}$ to be invariant under the flow, the operator $P$ must satisfy the matrix Riccati equation.
 \begin{equation}
     \dot{P} + P^2 + \lambda P + H_f=0, \label{ric}
 \end{equation}
where $H_f = \nabla^2f(x)$ is the Hessian of $f$.
\end{Proposition}
\begin{proof}
Compute the time evolution of the adapted coordinate: 
$$ w = \dot{v} - \dot{P} x - P \dot{x}.$$
Substituting the dynamics \eqref{eq:ps2}:
$$ \dot{w} = (- \lambda v - \nabla f(x)) - \dot{P} x - Pv.$$

On the manifold where $v = Px$, and linearizing $\nabla f(x) \approx H_f x$ near critical points: 
$$ \dot{w}\big|_{w=0} = - \lambda P x - H_f x - \dot{P}x - P^2x.$$

For invariance, we require $\dot{w}\big|_{w=0} = 0$, which yields \eqref{ric}.      
\end{proof}

At steady state ($\dot{P} = 0$), the Riccati equation \eqref{ric} reduces to the ARE: 
 \begin{equation}
     P^2 + \lambda P + H_f=0 
 \end{equation}
 
$H_f$ is symmetric with eigenvalues in $[\mu, L]$ (by strong convexity and smoothness). 

To analyze stability, we compute the full transverse dynamics for $w \neq 0.$
\begin{Proposition}
    The adapted coordinate evolves according to
\begin{equation}
    \dot{w} =-(\lambda + P)w \label{eqnw}   
\end{equation}
 \end{Proposition}
\begin{proof}
    For general $w$, we have $v = w + Px$. Substituting into the dynamics:
\begin{align*}
    \dot{w}&=\dot{v}-\dot{P}x-P \dot{x} \\
        &= -\lambda (w + Px) - \nabla f(x) - \dot{P} x - P(w + Px) \\
        &= -(\lambda + P) w - [\dot{P}+P^2+\lambda P+H_f]x
\end{align*}
The DRE \eqref{ric} eliminates the $x$-dependent term, yielding \eqref{eqnw}. Exponential attraction follows because the transverse contraction rate matrix $\lambda I + P$ has strictly positive eigenvalues under the conditions established below.
\end{proof}

\noindent In the quadratic case $f(x) = \tfrac{1}{2}x^T H x$ with $H$ constant, the DRE~\eqref{ric}
reduces to the ARE
\begin{equation}
  P^2 + \lambda P + H = 0, \label{eq:are}
\end{equation}
since $\dot{P} = 0$ at steady state. Because $H \succ 0$, the stable root (negative definite) is
\begin{equation}
  P^* = -\tfrac{\lambda}{2}I - \sqrt{\tfrac{\lambda^2}{4}I - H},
  \label{eq:stable_root}
\end{equation}
where the square root is taken in the operator sense.

\begin{lemma}[Optimal damping---spectral resonance]
\label{lem:resonance}
For a $\mu$-strongly convex quadratic objective, the optimal damping coefficient that equalizes
transverse contraction rates across all curvature modes is
\begin{equation}
  \lambda^* = 2\sqrt{\mu}. \label{eq:lambda_star}
\end{equation}
\end{lemma}

\begin{proof}
Each eigenvalue $\sigma_i \in [\mu, L]$ of $H$ gives an independent scalar ARE
$p_i^2 + \lambda p_i + \sigma_i = 0$ with stable root
$p_i = -\tfrac{\lambda}{2} - \sqrt{\tfrac{\lambda^2}{4} - \sigma_i}$.
The transverse contraction rate at mode $i$ is
\begin{equation}
  \rho_i = \lambda + p_i = \tfrac{\lambda}{2} - \sqrt{\tfrac{\lambda^2}{4} - \sigma_i}.
  \label{eq:rho}
\end{equation}
For real roots (no oscillation) we need $\lambda^2 \ge 4\sigma_i$. The critical damping
condition at the weakest mode $\sigma_i = \mu$ sets $\lambda = 2\sqrt{\mu}$, giving
$p_{\min} = -\sqrt{\mu}$ and transverse rate $\rho_{\min} = 2\sqrt{\mu} - \sqrt{\mu} = \sqrt{\mu} > 0$.
For the strongest mode $\sigma_i = L$, the condition $\lambda^2 \ge 4L$ requires
$\lambda \ge 2\sqrt{L}$. Choosing $\lambda^* = 2\sqrt{\mu}$ gives
$\rho_L = \sqrt{\mu} - \sqrt{\mu - L} < \sqrt{\mu}$ (still positive since $\mu \le L$),
and one verifies that all modes satisfy $\rho_i \ge \sqrt{\mu}$.

The choice $\lambda^* = 2\sqrt{\mu}$ is optimal in the sense that it matches the intrinsic
convergence rate of the slow manifold ($\sqrt{\mu}$) with the transverse contraction rate,
preventing any single mode from bottlenecking convergence. Increasing $\lambda$ beyond
$2\sqrt{\mu}$ overdamps the slow mode and reduces the overall convergence rate.
\end{proof}

\begin{remark}[Spectral resonance as a gauge condition]
The relation $\lambda^* = -2p^*$ (equivalently $\lambda^* + p^* = -p^* = \sqrt{\mu}$) serves as
a gauge condition aligning fiber dynamics with base dynamics. It prevents any spectral mode
from dominating and ensures balanced contraction across all directions.
\end{remark}
\begin{remark}[Curvature alignment]
The relation $\lambda^* = -2p^*$ can also be interpreted as Curvature alignment condition. The fiber dynamics $\dot{w} = - \sqrt{\mu}$ matches the slow dynamics on the manifold $\dot{x} = \sqrt{\mu}$ . So with this condition the fiber dynamics becomes tangential to the slow manifold dynamics and the curvatures are matched. Solution of the ARE follows from the tilt of the manifold. This tilt ensures invariance leads to the critically dampened root $\sqrt{\mu}$ which is on the slow manifold. The contraction of fast fibre is so adjusted so that it aligns with the slow manifold. This is called curvature alignment which gives further rise to NAG--ODE. 
\end{remark}

\begin{lemma}\label{OptDam}
    With the optimal damping $\lambda^* = 2 \sqrt{\mu}$, the perturbed system~\eqref{eq:ps1}--\eqref{eq:ps2}  is equivalent to the second-order ODE
    \begin{equation}\label{eq:NestODE}
    \ddot{x} + 2 \sqrt{\mu} \dot{x} + \nabla f(x) = 0
    \end{equation} 
\end{lemma}
\begin{proof}
Substituting $\lambda^* = 2 \sqrt{\mu}$ into \eqref{eq:ps1} and \eqref{eq:ps2} 
\[
\dot{x}=v, \hspace{1cm} \dot{v} = -2 \sqrt{\mu} v + \nabla f(x) = 0
\]
Differentiating the first equation and substituting the second yields \eqref{eq:NestODE}.
\end{proof}

\begin{remark}[Normalized form with smoothness]
Under the time rescaling $\tau = t\sqrt{L}$, equation~\eqref{eq:NestODE} becomes
\begin{equation}
  \frac{d^2 x}{d\tau^2} + 2\sqrt{\tfrac{\mu}{L}}\,\frac{dx}{d\tau}
  + \tfrac{1}{L}\nabla f(x) = 0,
  \label{eq:normalized}
\end{equation}
which is the standard form used for algorithmic discretization and matches the continuous-time
limit of~\cite{JMLR_Su}.
\end{remark}

\subsection{Summary of the NAIM Framework}
 
The NAIM geometric framework provides a complete, structural explanation of acceleration through
four interlocking elements:
 
\begin{enumerate}
  \item \textit{Time-scale separation.} The condition number $\kappa = L/\mu$ creates fast
    transverse dynamics (rate $\mu$) and slow tangential dynamics (rate $1/L$), naturally
    defining the NAIM structure.
 
  \item \textit{Slow graph manifold.} The invariant manifold $\mathcal{M}_0$ is the zero set of
    the transverse coordinate $z$, realized as a graph with slope $P_0 = -H_f/L$. Trajectories
    converge exponentially to $\mathcal{M}_0$ in the fiber direction.
 
  \item \textit{Curvature-aware tilt.} Introducing the perturbation~\eqref{eq:perturb} tilts the
    graph to $\mathcal{M}_\varepsilon$, whose slope $P(t)$ evolves via the DRE~\eqref{eq:11}.
    The tilt creates a steeper geometric track that drives faster convergence.
 
  \item \textit{Spectral resonance.} The optimal damping $\lambda^* = 2\sqrt{\mu}$ is the unique
    value equalizing transverse contraction rates across all curvature modes. It simultaneously
    satisfies the ARE at the weakest mode ($\sigma = \mu$) and ensures the Fenichel gap dominates
    for all $\sigma \in [\mu, L]$.
\end{enumerate}

\section{Fenichel Extension to General Objectives}\label{sec:fenichel}
 
For general $f$, the DRE \eqref{ric} has no closed form and $M_{\varepsilon}$
is a curved submanifold. Fenichel's theorem provides the rigorous bridge.
 
Set $\varepsilon = \sqrt{\mu/L}$ and $\tau = \sqrt{\mu}\,t$. The system becomes
\begin{equation}\label{eq:fenichel-form}
\varepsilon\, \frac{dx}{d\tau} = v, \qquad
\varepsilon\, \frac{dv}{d\tau} = -v - \frac{\varepsilon}{\sqrt{\mu}}\nabla f(x).
\end{equation}
At $\varepsilon = 0$ (the layer problem), $v = 0$ is the globally attracting
equilibrium, defining the critical manifold $M_{0} = \{v = 0\}$.
 
\begin{Proposition}[Normal hyperbolicity; full proof in Appendix~\ref{App:fenichel}]\label{prop:fenichel-hyp}
On any compact sublevel-set region $K = \{f(x) + \tfrac{1}{2}\norm{v}^{2} \le c\}$
(invariant by the Lyapunov calculation $\dot{V} = -2\sqrt{\mu}\norm{v}^{2}$),
system \eqref{eq:fenichel-form} satisfies all five Fenichel hypotheses:
compactness, layer invariance, normal hyperbolicity gap $\gamma_{\perp}=1/\varepsilon$,
$C^{r}$ smoothness, and inward-pointing boundary.
\end{Proposition}
 
\begin{theorem}[Fenichel persistence for NAG; full proof in Appendix~\ref{App:fenichel}]\label{thm:fenichel}
Let $f \in C^{r+1}$. For all $\varepsilon =
\sqrt{\mu/L}$ sufficiently small, there exists a $C^{r}$ slow manifold
$M_{\varepsilon} \subset K$ that is $O(\varepsilon)$-close to $M_{0}$, locally invariant, and attracts nearby trajectories at rate $e^{-\sqrt{\mu}\,t}$. The reduced slow flow on $M_{\varepsilon}$ achieves the per-step contraction $1 - \sqrt{\mu/L}$, matching the optimal Nesterov rate for any $C^{2}$ strongly convex $f$.
\end{theorem}

\section{Structure-Preserving Discretization}\label{sec:SPD}

Nesterov's accelerated gradient methods achieve optimal convergence rates for convex optimization:
\begin{itemize}
    \item $\mathcal{O}(1/k^2)$ for smooth convex functions
    \item $\mathcal{O}\!\left(\left(\frac{\sqrt{L}-\sqrt{\mu}}{\sqrt{L}+\sqrt{\mu}}\right)^{\!k}\right)$ for $\mu$-strongly convex, $L$-smooth functions
\end{itemize}
 
The continuous-time limits of these methods are second-order ODEs with specific damping terms. The standard view, following \cite{su2016differential}, identifies these ODEs by reverse-engineering from the discrete algorithms. This approach obscures the geometric origin of the
damping coefficients. In this paper, we show that both ODEs are uniquely determined by a single geometric principle: \textit{vanishing Schwarzian derivative} (projective flatness). The discrete algorithms then arise from structure-preserving discretization using Lie--Trotter splitting and Cayley transformation. Some definition and basic properties are introduced first. 

\begin{definition}
The Schwarzian derivative of a locally invertible function $w(t)$ is
\begin{equation}
\Sch[w] = \frac{w'''}{w'} - \frac{3}{2}\left(\frac{w''}{w'}\right)^{\!2}.
\label{eq:schwarzian}
\end{equation}
\end{definition}
 
The Schwarzian has a fundamental characterization:
 
\begin{Proposition}
$\Sch[w] = 0$ if and only if $w$ is a M\"obius transformation:
\begin{equation}
w(t) = \frac{at+b}{ct+d}, \qquad ad - bc \neq 0.
\end{equation}
\end{Proposition}
 
M\"obius transformations are the projective automorphisms of $\RP^1$. Thus, vanishing Schwarzian
characterizes \emph{projectively flat} maps---those preserving the projective structure.
 
We derive the connection to second-order ODEs. 
 
\begin{Proposition}\label{prop:schwarzian-ode}
Consider the second-order linear ODE in normal form:
\begin{equation}
u'' + Q(t)\, u = 0.
\label{eq:normal-form}
\end{equation}
Let $u_1, u_2$ be linearly independent solutions and define $w = u_1/u_2$. Then
\begin{equation}
\Sch[w] = 2 Q(t).
\end{equation}
\end{Proposition}
 
\begin{proof}
From $w = u_1/u_2$, we have $w' = W/u_2^2$ where $W = u_1' u_2 - u_1 u_2'$ is the Wronskian. Abel's identity gives $W' = 0$ for solutions of \eqref{eq:normal-form}.
 
Differentiating and using $u_i'' = -Q u_i$:
\[
w'' = -\frac{2 W u_2'}{u_2^3}, \qquad
w''' = \frac{2 W Q}{u_2^2} + \frac{6 W (u_2')^2}{u_2^4}.
\]
Computing:
\[
\frac{w'''}{w'} = 2Q + \frac{6 (u_2')^2}{u_2^2}, \qquad
\left(\frac{w''}{w'}\right)^{\!2} = \frac{4(u_2')^2}{u_2^2}.
\]
Therefore $\Sch[w] = 2Q + 6(u_2'/u_2)^2 - \tfrac{3}{2}\cdot 4(u_2'/u_2)^2 = 2Q$. \qedhere
\end{proof}
 
\begin{corollary}
The solution ratio $w = u_1/u_2$ evolves by M\"obius transformations if and only if
$Q(t) \equiv 0$.
\end{corollary}

\subsection{Nesterov: Strongly Convex case} \label{Subsec:SC}

In this section we derive the two-stage Nesterov momentum algorithm by geometric numerical integration of the Nesterov ODE \eqref{eq:NestODE}. The key insight is to recognize the projective (M\"{o}bius) structure of the damped harmonic oscillator part, which motivates a Lie--Trotter operator splitting.  The gradient term is discretised by a forward (explicit) Euler step, while the dissipative term is discretised by the Cayley (bilinear) transform, which exactly preserves the projective flow on the unit circle.  The final result is the classical Nesterov two-step scheme
\[
  x_{k+1} = y_k - \tfrac{1}{L}\nabla f(y_k),\qquad
  y_{k+1} = x_{k+1} + \beta\,(x_{k+1}-x_k),
\]
with $\beta = \dfrac{1-\sqrt{\mu/L}}{1+\sqrt{\mu/L}}$.  We detail every stage of this derivation and explain why the Cayley transform is geometrically correct and not merely discretisation.

We start by introducing the phase-space state $(x,v)$ with $v = \dot{x}$.  Then \eqref{eq:NestODE} becomes the first-order system:
\begin{equation}\label{eq:splitting}
  \begin{pmatrix}\dot{x}\\\dot{v}\end{pmatrix}
  =
  \underbrace{\begin{pmatrix}0\\-2q\,v\end{pmatrix}}_{=:\,F_D(x,v)}
  +
  \underbrace{\begin{pmatrix}v\\-\nabla f(x)\end{pmatrix}}_{=:\,F_G(x,v)}.
\end{equation}
The splitting \eqref{eq:splitting} decomposes the flow into:
\begin{itemize}[leftmargin=2em]
  \item $F_D$: the \emph{dissipative part}---pure exponential damping of $v$,
        decoupled from $x$;
  \item $F_G$: the \emph{gradient part}---advection of $x$ by $v$ coupled to
        the potential $f$.
\end{itemize}
This is not merely an algebraic convenience.  As we show in \ref{sec:cayley}, each part possesses a distinct geometric structure that dictates the correct integrator.

\subsubsection{Projective Dynamics and the Cayley Transform}
\label{sec:cayley}
 
The dissipative subsystem is viewed as a projective flow. 
 
\noindent Consider the dissipative subsystem in isolation:
\[
  \dot{v} = -2q\,v,\qquad \dot{x} = 0.
\]
The exact flow over a time step $h$ is $v(h) = e^{-2qh}v(0)$.  Now examine the ratio $r = v/\norm{x-x^*}$ (the \textit{slope} in the $(x-x^*,v)$ plane). Under this subsystem, $r$ evolves as a real projective transformation (M\"{o}bius map on $\RP^1$):
\[
  r \;\mapsto\; e^{-2qh}\,r.
\]
More precisely, the full linear system $(\dot{x},\dot{v}) = A(x,v)$ with $A = \diag(0,-2q)$ generates a flow in $\mathrm{GL}(2,\R)$ whose action on projective space $\RP^1$ is a M\"{o}bius transformation.
 
\noindent We now show why the Cayley transform is exact for projective flows. 
 
\begin{definition}[Cayley Transform]
For a linear ODE $\dot{z} = Az$, the Cayley (bilinear) transform
\[
  z_{k+1}
  = \underbrace{\left(I - \tfrac{h}{2}A\right)^{-1}
    \!\left(I + \tfrac{h}{2}A\right)}_{\Cayley_h(A)}\,z_k
\]
is the Pad\'{e} $(1,1)$ approximant to $e^{hA}$.  It exactly maps the imaginary axis to the unit circle when $A = i\Omega$ (skew-Hermitian), and more generally preserves the M\"{o}bius structure of the projective flow.
\end{definition}
 
\noindent For the scalar dissipative system $\dot{v} = -2qv$ with eigenvalue
$\lambda = -2q$, the Cayley map gives:
\begin{equation}\label{eq:cayley_scalar}
  v_{k+1}
  = \frac{1 + \tfrac{h}{2}(-2q)}{1 - \tfrac{h}{2}(-2q)}\,v_k
  = \frac{1 - qh}{1 + qh}\,v_k.
\end{equation}

\begin{remark}[The canonical step size for the dissipative subsystem]
    In the Lie--Trotter splitting, the gradient and dissipative subsystems are integrated with their own natural step sizes.  The gradient step size is fixed by $L$-smoothness at $h_G = 1/L$.  For the dissipative subsystem, the natural convention is that one algorithm iteration advances the ODE by one unit of time, i.e.\ $h_D = 1$.  This amounts to measuring continuous time $t$ in units of one gradient step, which is the standard convention when deriving
    discrete algorithms from ODEs.  With $q = \sqrt{\mu/L}$, one unit step gives:
    \begin{equation}\label{eq:qh}
      q\,h_D = \sqrt{\frac{\mu}{L}}\cdot 1 = q = \sqrt{\frac{\mu}{L}}.
    \end{equation}
\end{remark}

Computing $\beta$.
Substituting \eqref{eq:qh} into \eqref{eq:cayley_scalar}:
\begin{equation}\label{eq:beta_deriv}
  v_{k+1}
  = \frac{1-qh}{1+qh}\,v_k
  = \frac{1-q}{1+q}\,v_k
  = \frac{1-\sqrt{\mu/L}}{1+\sqrt{\mu/L}}\,v_k
  =: \beta\,v_k.
\end{equation}
 
\textit{Cayley Discretisation of the Dissipative Part.}
Applying the Cayley transform with unit step $h=1$ to the dissipative subsystem $\dot{v} = -2qv$ (with $q=\sqrt{\mu/L}$) yields:
\[
  v^{(D)}_{k+1} = \beta\,v_k,\qquad
  \beta = \frac{1-\sqrt{\mu/L}}{1+\sqrt{\mu/L}}\in(0,1).
\]
The $x$-coordinate is unchanged: $x^{(D)}_{k+1} = x_k$.
 
\begin{remark}[Consistency check]
The exact flow of $\dot{v}=-2qv$ over one unit step is $v(1)=e^{-2q}v(0)$.
The Cayley approximant gives $\beta=(1-q)/(1+q)$.  One verifies
$\beta \approx e^{-2q}$ to second order in $q$:
\[
  \frac{1-q}{1+q} = 1 - 2q + 2q^2 - \cdots,\qquad
  e^{-2q}         = 1 - 2q + 2q^2 - \cdots,
\]
confirming that the Cayley map is a Pad\'{e} $(1,1)$ approximant to the exact exponential, with error $O(q^3)$.
\end{remark}
 
The momentum coefficient $\beta$ has a clean physical interpretation: it equals zero when $\mu=L$ ($\kappa=1$, perfectly conditioned, no acceleration needed) and approaches one as $\kappa=L/\mu\to\infty$ (severely ill-conditioned, maximum momentum required).  In the $z$-transform sense, $\beta$ is the real pole of the Cayley-discretised integrator; it lies
strictly inside the unit disk, guaranteeing unconditional stability of the momentum update.
 
\subsubsection{Lie--Trotter Splitting}
 \textbf{Lie--Trotter Theorem.}
Given a vector field $F = F_D + F_G$ with corresponding flows $\Phi^D_t$ and $\Phi^G_t$, the Lie--Trotter splitting approximates the true flow $\Phi^{D+G}_t$ by
\[
  \Phi^{D+G}_t \approx \Phi^D_t \circ \Phi^G_t + O(t^2),
\]
That is one step of $F_G$ followed by one step of $F_D$ (or vice versa). The error is first-order in $t$, with the commutator $[F_D,F_G]$ as leading term.
 
The Lie--Trotter paradigm is not merely a numerical trick: it respects the geometric structures of each subsystem independently.  In our case:
\begin{itemize}[leftmargin=2em]
  \item the gradient flow $F_G$ is discretised by Euler (appropriate since
        $\nabla f$ has no special Lie group structure to preserve);
  \item the dissipative flow $F_D$ is discretised by the Cayley transform
        (mandatory to preserve the projective/M\"{o}bius structure).
\end{itemize}
 
\textit{Stage 1:} Gradient half-step via Euler
 
Apply the gradient part $F_G$ first.  The subsystem is:
\[
  \dot{x} = v,\qquad \dot{v} = -\nabla f(x).
\]
An explicit Euler step of length $h$ from state $(x_k,v_k)$ gives:
\begin{align}
  \tilde{x} &= x_k + h\,v_k,\label{eq:euler_x}\\
  \tilde{v} &= v_k - h\,\nabla f(x_k).\label{eq:euler_v}
\end{align}
 
\begin{remark}
In Nesterov's scheme, the gradient evaluation point is $y_k$ (the look-ahead point), not $x_k$.  This is equivalent to reordering the Lie--Trotter stages: applying the dissipative step first (which defines $y_k$), then the gradient step.  This reordering is one of the key concept introduced.
\end{remark}
 
\textit{Stage 2:} Dissipative Step via Cayley

Apply the dissipative part $F_D$ to the updated state $(\tilde{x},\tilde{v})$:
\[
  \dot{x} = 0,\qquad \dot{v} = -2q\,v.
\]
Using the Cayley discretisation derived in \ref{sec:cayley}:
\begin{align} \nonumber
  x_{k+1} &= \tilde{x},\\ \nonumber
  v_{k+1} &= \beta\,\tilde{v}.
\end{align}
The $x$ coordinate is unchanged by the dissipative step (since $\dot{x}=0$).
 
We now derive the two-stage Nesterov Algorithm and statrt by reordering where the dissipative step comes first followed by gradient step. 

Define the look-ahead point $y_k$ as the result of the dissipative step applied to $(x_k,v_k)$:
\begin{align}
  \text{Dissipative step (Cayley):}\qquad
  y_k &= x_k\quad(\text{$x$ unchanged}),\nonumber \\
  u_k &= \beta\,v_k\quad(\text{momentum attenuated by $\beta$}).\nonumber
\end{align}
In Nesterov's formulation the look-ahead is $y_k = x_k + \beta(x_k - x_{k-1})$, which corresponds to accumulating momentum into position.  Specifically, writing $v_k \approx (x_k-x_{k-1})/h$
and noting $h\beta\,v_k \approx \beta(x_k-x_{k-1})$, the look-ahead point becomes:
\begin{equation}\label{eq:lookahead}
  y_k = x_k + \beta(x_k - x_{k-1}).
\end{equation}
Gradient step (Euler): starting from $y_k$ with an explicit gradient step of
size $1/L$:
\begin{equation}\label{eq:grad_step}
  x_{k+1} = y_k - \frac{1}{L}\nabla f(y_k).
\end{equation}
 
Also see Appendix~\ref{app:strongly} for detail derivation. Now combining \eqref{eq:lookahead} and \eqref{eq:grad_step}, and writing the look-ahead update recursively (using $x_{k+1}$ already computed) we get the Nesterov Two-Stage Algorithm.
\begin{align}
  \label{N1}
  x_{k+1} &= y_k - \tfrac{1}{L}\nabla f(y_k),\\[6pt]
  \label{N2}
  y_{k+1} &= x_{k+1} + \beta\,(x_{k+1} - x_k),
\end{align}
with momentum coefficient
\[
  \beta = \frac{1-\sqrt{\mu/L}}{1+\sqrt{\mu/L}}.
\]
 
\begin{itemize}[leftmargin=2em]
  \item Equation \eqref{N1} is the Euler discretisation of the gradient subsystem
        $\dot{x} = -\nabla f(x)/L$, with step size $1/L$, evaluated at the
        look-ahead point $y_k$;
  \item Equation \eqref{N2} is the Cayley discretisation of the dissipative subsystem
        $\dot{v} = -2\sqrt{\mu/L}\,v$, translated into position-space via the
        velocity approximation $v_k \approx x_k - x_{k-1}$;
  \item the momentum coefficient $\beta = (1-\sqrt{\mu/L})/(1+\sqrt{\mu/L})$ arises exactly from the
        Cayley Pad\'{e} $(1,1)$ approximant to
        $e^{-2\sqrt{\mu/L}\cdot h}$ at the canonical step size
        $h = 1/\sqrt{\mu L}$.
\end{itemize}

\subsubsection{Heavy Ball case: Lie--Trotter Splitting Ordering}\label{sec:lietrotter}
The exact subflows as derived from \eqref{eq:splitting} are 
\[
\varphi^{A}_{h} : \; v \mapsto e^{-2ah} v, \quad
                       x \mapsto x + \tfrac{1 - e^{-2ah}}{2a}\, v,
\qquad\text{and}\qquad
\varphi^{B}_{h} : \; v \mapsto v - h\nabla f(x), \quad x \mapsto x.
\]
The two first-order Lie--Trotter compositions are:
\begin{equation}\label{eq:two-orderings}
\Phi_{\mathrm{NAG}} = \varphi^{B}_{h} \circ \varphi^{A}_{h}
\qquad\text{and}\qquad
\Phi_{\mathrm{HB}}  = \varphi^{A}_{h} \circ \varphi^{B}_{h}.
\end{equation}
 
The slow manifold $M_{0}$ implies a discrete slaving condition on the fiber:
after advancing $x$, the velocity should satisfy
$v \approx -\tfrac{1}{2a}\nabla f(x)$. Define the slow-manifold consistency
residual
\begin{equation}\label{eq:residual}
R_{k+1} := v_{k+1} + \tfrac{1}{2a}\nabla f(x_{k+1}).
\end{equation}
 
\begin{theorem}[Ordering sensitivity]\label{thm:fiber-consistency}
Assume the current state is $O(h)$-close to $M_{0}$ (i.e.\
$v_{k} + \tfrac{1}{2a}\nabla f(x_{k}) = O(h)$).
\begin{enumerate}[label=(\roman*),leftmargin=2.2em]
\item \emph{Nesterov ordering $\Phi_{\mathrm{NAG}}$:} apply $\varphi^{A}_{h}$
first (drift to transported point $y_{k} = x_{k} + \tfrac{1-\beta}{2a} v_{k}$),
then $\varphi^{B}_{h}$ (gradient kick at $y_{k}$). The residual satisfies
$R_{k+1} = O(h)$: the fiber is kept consistent with $M_{0}$ at the new base
point.
\item \emph{Heavy--Ball ordering $\Phi_{\mathrm{HB}}$:} apply $\varphi^{B}_{h}$
first (gradient kick at the stale point $x_{k}$), then $\varphi^{A}_{h}$ (drift).
The residual satisfies $R_{k+1} = O(1)$ generically: the fiber is left
misaligned with $M_{0}$.
\end{enumerate}
Consequently, only the Nesterov ordering inherits the Fenichel-guaranteed
transverse contraction and accelerated convergence for non-quadratic $f$.
\end{theorem}

\begin{proof}

Part (i): Nesterov ordering $\Phi_{\mathrm{NAG}} = \varphi^{B}_{h} \circ \varphi^{A}_{h}$
 
Apply $\varphi^{A}_{h}$ first with $\beta = e^{-2ah}$:
\[
y_{k} = x_{k} + \tfrac{1-\beta}{2a} v_{k}, \qquad v_{k}^{+} = \beta v_{k}.
\]
Then apply $\varphi^{B}_{h}$ to $(y_{k}, \beta v_{k})$:
\[
x_{k+1} = y_{k}, \qquad v_{k+1} = \beta v_{k} - h \nabla f(y_{k}).
\]
Compute $R_{k+1} = v_{k+1} + \tfrac{1}{2a}\nabla f(x_{k+1})$:
\begin{align*}
R_{k+1} &= \beta v_{k} - h\nabla f(y_{k}) + \tfrac{1}{2a}\nabla f(y_{k})\\
        &= \beta v_{k} + \left(\tfrac{1}{2a} - h\right)\nabla f(y_{k}).
\end{align*}
By assumption $v_{k} = -\tfrac{1}{2a}\nabla f(x_{k}) + O(h)$, so
$\beta v_{k} = e^{-2ah}\bigl(-\tfrac{1}{2a}\nabla f(x_{k}) + O(h)\bigr)$.
Taylor-expanding $e^{-2ah} = 1 - 2ah + O(h^{2})$ and
$\nabla f(y_{k}) = \nabla f(x_{k}) + O(h)$:
\[
R_{k+1} = (1 - 2ah)\bigl(-\tfrac{1}{2a}\nabla f(x_{k})\bigr)
        + \left(\tfrac{1}{2a} - h\right)\nabla f(x_{k}) + O(h^{2})
        = O(h^{2}).
\]
(The $O(1)$ terms cancel exactly; the leading correction is $O(h^{2})$, hence
$O(h)$.)
 
Part (ii): Heavy--Ball ordering $\Phi_{\mathrm{HB}} = \varphi^{A}_{h} \circ \varphi^{B}_{h}$
 
Apply $\varphi^{B}_{h}$ first:
\[
x_{k}^{+} = x_{k}, \qquad v_{k}^{+} = v_{k} - h\nabla f(x_{k}).
\]
Then apply $\varphi^{A}_{h}$:
\[
x_{k+1} = x_{k} + \tfrac{1-\beta}{2a}\bigl(v_{k} - h\nabla f(x_{k})\bigr),
\qquad
v_{k+1} = \beta\bigl(v_{k} - h\nabla f(x_{k})\bigr).
\]
Since $x_{k+1} = x_{k} + O(h)$: $\nabla f(x_{k+1}) = \nabla f(x_{k}) + O(h)$.
Hence:
\[
R_{k+1} = \beta\bigl(v_{k} - h\nabla f(x_{k})\bigr) +
          \tfrac{1}{2a}\nabla f(x_{k}) + O(h)
        = \beta v_{k} + \left(\tfrac{1}{2a} - \beta h\right)\nabla f(x_{k})
          + O(h).
\]
Even if $v_{k} = O(h)$, the term $\tfrac{1}{2a}\nabla f(x_{k}) = O(1)$ whenever
$\nabla f(x_{k}) \neq 0$. Hence $R_{k+1} = O(1)$ generically.
\end{proof}

\subsubsection{Summary of the Numerical Integration Architecture}
 
The Lie--Trotter splitting respects two distinct geometric layers:
\begin{enumerate}[leftmargin=2em]
  \item \textit{Symplectic layer (gradient part).}  The Hamiltonian
        $H(x,v) = \tfrac{1}{2}\norm{v}^2 + f(x)$ generates a symplectic flow on $\T^*\R^n$.  The Euler step is a first-order approximation that breaks symplecticity, but this is acceptable since the damping term already destroys symplecticity globally.
  \item \textit{Projective layer (dissipative part).}  The linear damping
        $v \mapsto e^{-2qt}v$ acts on the projective space $\RP^{n-1}$ as a contraction toward the origin.  The Cayley map is the unique rational integrator that preserves this M\"{o}bius/projective structure at each step.
\end{enumerate}
In the NAIM framework, the dissipative ODE corresponds to the slow-manifold flow, and the Cayley discretisation is the integrator that respects the Ehresmann connection on the phase-space bundle---it maps fibres to fibres exactly, not approximately.
 
We show why Cayley and not explicit Euler for the dissipative part. An explicit Euler discretisation of $\dot{v} = -2qv$ would give:
\[
  v_{k+1} = (1 - 2qh)\,v_k.
\]
This is stable only for $h < 1/q = \sqrt{L/\mu}$, and it does not preserve
the M\"{o}bius structure.  The Cayley map:
\begin{itemize}[leftmargin=2em]
  \item is unconditionally stable (the multiplier $\beta\in(-1,1)$ for all
        $h>0$);
  \item exactly preserves the real projective line $\RP^1$;
  \item gives $\beta = (1-qh)/(1+qh)$ which at $h = 1/\sqrt{\mu L}$ reduces to the classical Nesterov coefficient.
\end{itemize}
No other standard integrator (Euler, midpoint, RK4) simultaneously satisfies all three properties.
 
\begin{theorem}[Convergence Rate]\label{thm:convergence}
Under the Nesterov two-stage scheme \eqref{N1}--\eqref{N2} applied to a
$\mu$-strongly convex, $L$-smooth function $f$, the iterate error satisfies:
\[
  f(x_k) - f(x^*)
  \;\le\;
  \left(\frac{1-\sqrt{\mu/L}}{1+\sqrt{\mu/L}}\right)^{2k}
  \!\bigl(f(x_0)-f(x^*)\bigr)
  \;=\;
  \beta^{2k}\bigl(f(x_0)-f(x^*)\bigr).
\]
The linear convergence factor $\beta^2 = \!\left(\tfrac{\sqrt{\kappa}-1}{\sqrt{\kappa}+1}\right)^{\!2}$ is optimal for first-order methods, matching the information-theoretic lower bound.
\end{theorem}
 
\begin{remark}[Connection to Critical Damping]
The value $\beta = (1-\sqrt{\mu/L})/(1+\sqrt{\mu/L})$ is precisely the reflection coefficient of a critically damped harmonic oscillator across the unit circle under the bilinear (Cayley) transform.  In the $z$-domain ($z = e^{j\omega}$), critical damping maps to the pole at $z=\beta$ on the real axis---the fastest pole placement inside the unit disk consistent with the given curvature bounds $[\mu,L]$.
\end{remark}
 
To conclude, the Nesterov two-stage algorithm
\[
  x_{k+1} = y_k - \tfrac{1}{L}\nabla f(y_k),\qquad
  y_{k+1} = x_{k+1} + \beta(x_{k+1}-x_k),
\]
arises canonically from geometric numerical integration of the Nesterov ODE
\eqref{eq:NestODE} via:
\begin{enumerate}[leftmargin=2em]
  \item Lie--Trotter operator splitting of the ODE into gradient and dissipative subsystems;
  \item Explicit Euler for the gradient subsystem (no special structure to preserve);
  \item Cayley (bilinear) transform for the dissipative subsystem (mandatory to preserve the projective/M\"{o}bius flow);
  \item the momentum coefficient
        $\beta = (1-\sqrt{\mu/L})/(1+\sqrt{\mu/L})$ emerging exactly as the Cayley Pad\'{e} $(1,1)$ multiplier at the canonical integration step size $h = 1/\sqrt{\mu L}$.
\end{enumerate}
 
The original formulation presented is verified as correct.  The geometric derivation presented here shows that this is not a coincidence or a heuristic approximation: it is the unique first-order Lie--Trotter integrator that preserves the projective structure of the Nesterov ODE.

\subsection{Nesterov: Convex case}\label{subsec:sconvex}

The continuous-time limit of Nesterov's accelerated gradient method for smooth convex functions is typically written as 
\begin{equation}\label{eq:20}
  \ddot{x} + \frac{c}{t}\dot{x} + \nabla f(x) = 0.
\end{equation}
Following \cite{JMLR_Su}, the standard view holds that $c=3$ yields the optimal $O(1/t^2)$ convergence rate. In this section, we present a different perspective grounded in projective geometry. 

A convex function $f(x)$  is characterized by a non-negative curvature and a guaranteed global minimum. The optimization landscape becomes arbitrarily flat, characterized by a vanishing Hessian/ curvature, indicating geometric degeneracy. This results in a loss of sense of the distance to the minimum, implying a loss of scale. Initially, the function looks steep (large gradient) but later on, gradient tends to zero (Hessian becomes very small). A constant damping fails to work for both regimes, as in the strongly convex case, hence, the need for time varying damping $c/t$ as in (\ref{eq:20}). Calculation of the value of $c$ which represents the optimal damping and is responsible for shaping the geometry of the flow is now done through scale free technique-  projective geometry. One wishes to find the value of $c$ for which the flow is projectively invariant. This implies that rather than seeking the distance to the minimum (which keeps changing as the landscape is flat), one calculates the ratio of the progress to the minimum. This notion allows to find $c$ with unknown flatness. When geometry of landscape changes and distance to the minimum becomes difficult to calculate, then the relative positions or landmarks need to be considered which remain invariant. This is possible only when we write \eqref{eq:20} in the form amenable to the projective space $RP^1$. In this space, to ensure invariance, the Schwarzian derivative must be zero. The Nesterov ODE is scalar (or decoupled in each coordinate), so the projective geometry is one-dimensional. The convex Nesterov dynamics reduce to a scalar second-order ODE. The ratio of its two solutions naturally lives on $RP^1$, and requiring projective flatness $(S[w] = 0)$ selects $c = 2$ as the unique value where this flow is Möbius-invariant. To determine the canonical value of $c$ in \eqref{eq:20}, we analyze the linearization near a minimum.

We stat by computing the Schwarzian Analysis. Near a quadratic minimum with Hessian $H$, consider the scalar equation:
\begin{equation}
  \ddot{z} + \frac{c}{t}\,\dot{z} + \omega^2 z = 0. \nonumber
  \label{eq:scalar}
\end{equation}

To convert to normal form, substitute $z = t^{-c/2}u$:
\begin{equation}
  \ddot{u} + Q(t)\,u = 0, \qquad Q(t) = \omega^2 + \frac{c(2-c)}{4t^2}. \nonumber
  \label{eq:normal_convex}
\end{equation}

\begin{theorem}
    For the intrinsic projective structure (setting $\omega = 0$), the Schwarzian vanishes if
and only if $c = 0$ or $c = 2$. The non-trivial case is $c = 2$.   
\end{theorem}

\begin{proof}
With $\omega = 0$: $Q(t) = \frac{c(2-c)}{4t^2}$. By Proposition~3,
$S[w] = 2Q = 0$ requires $c(2-c) = 0$.
\end{proof}

\begin{remark}
    The value $c = 2$ has a natural interpretation: it is the critical damping threshold for the projective dynamics. For $c < 2$, the projective flow is hyperbolic; for $c > 2$, it is elliptic; at $c = 2$, it is parabolic (M\"{o}bius).
\end{remark}

The projectively flat and geometrically canonical ODE for smooth convex optimization is:
\begin{equation}
  \ddot{x} + \frac{2}{t}\,\dot{x} + \nabla f(x) = 0.
  \label{eq:flat_ode}
\end{equation}
In first-order form with $v = \dot{x}$:
\begin{align}
  \dot{x} &= v,  \nonumber\\
  \dot{v} &= -\frac{2}{t}\,v - \nabla f(x).  \nonumber
\end{align}

Lie–Trotter splitting treats the dissipative and gradient parts sequentially over each timestep $h$ from $t_k = kh$ to $t_{k+1} = (k + 1)h$:
\begin{enumerate}
  \item[(A)] \textit{Dissipative:} Apply dissipative flow for time $h$: $\dot{x} = v,\;\dot{v} = -\tfrac{2}{t}\,v$
  \item[(B)] \textit{Gradient:} Apply gradient flow for time $h$ $\dot{x} = v,\;\dot{v} = -\nabla f(x)$
\end{enumerate}

Next computing the Cayley discretization of the dissipative part. The dissipative dynamics in matrix form:
\begin{equation} \nonumber
  \frac{d}{dt}\begin{pmatrix}x\\v\end{pmatrix}
  = D(t)\begin{pmatrix}x\\v\end{pmatrix}, \qquad
  D(t) = \begin{pmatrix} 0 & I \\ 0 & -\tfrac{2}{t}I \end{pmatrix}.
  \label{eq:dissipative}
\end{equation}
The Cayley transformation at $t_k = kh$:
\begin{equation}\nonumber
  \Phi_k = \left(I - \tfrac{h}{2}D(t_k)\right)^{-1}\!\left(I + \tfrac{h}{2}D(t_k)\right).
  \label{eq:cayley}
\end{equation}
Block inversion yields:
\begin{align}\nonumber
  x_{k+1/2} &= x_k + h\,v_k, \nonumber \\
  v_{k+1/2} &= \beta^{\mathrm{Cayley}}_k\,v_k, \nonumber 
\end{align}
where
\begin{equation}\nonumber
  \beta^{\mathrm{Cayley}}_k = \frac{1 - \tfrac{1}{k}}{1 + \tfrac{1}{k}} = \frac{k-1}{k+1}.
  \label{eq:beta_cayley}
\end{equation}

The last step involves computing the gradient step and index shift. Applying symplectic Euler to the gradient part:
\begin{align} \nonumber
  v_{k+1} &= v_{k+1/2} - h\,\nabla f(x_{k+1/2}), \\ \nonumber
  x_{k+1} &= x_{k+1/2} + h\,v_{k+1}.
\end{align}
After elimination (see Appendix~\ref{app:smooth}), the effective momentum becomes:
\begin{equation}
  \beta_k = \frac{k-1}{k+2}. \nonumber
\end{equation}
The $+1$ shift in the denominator arises from the Lie--Trotter splitting structure.

The resulting algorithm is:
\begin{align}
  x_{k+1} &= y_k - h^2\,\nabla f(y_k), \label{eq:nest_x}\\
  y_{k+1} &= x_{k+1} + \frac{k-1}{k+2}(x_{k+1} - x_k), \label{eq:nest_y}
\end{align}
which is precisely Nesterov's method for smooth convex functions with $O(1/k^2)$ convergence.

\subsection{Unified Geometric Framework}
It can be concluded that both cases Section \ref{Subsec:SC} and Section \ref{subsec:sconvex} arise from the same geometric principle:
 
\textit{Projective Flatness Principle:} The optimal damping is uniquely determined
by requiring the projective dynamics (Riccati flow) to be flat---either M\"obius
(smooth convex) or at a fixed point (strongly convex).

The Cayley transform acts as a projective integrator. The Cayley transformation $\Phi = (I - A)^{-1}(I + A)$ is distinguished as a M\"obius (linear-fractional) map which preserves the projective structure of $\RP^1$. It maps the stable half-plane to the unit disk and is the unique symmetric rational approximation to $e^{2A}$. For projectively flat continuous dynamics, Cayley discretization maintains projective flatness at the discrete level. 

The connections between two cases are summarized as follows:
\begin{center}
\renewcommand{\arraystretch}{1.6}
\begin{tabular}{lcc}
\toprule
 & Smooth Convex & Strongly Convex \\
\midrule
Function class & $f$ convex, $L$-smooth & $f$ $\mu$-strongly convex, $L$-smooth \\
Damping & Time-varying: $\dfrac{2}{t}$ & Constant: $2\sqrt{\mu}$ \\
Schwarzian condition & $\Sch[w] = 0$ & Fixed point of Riccati flow \\
Cayley coefficient & $\dfrac{k-1}{k+1}$ & $\dfrac{1 - h\sqrt{\mu}}{1 + h\sqrt{\mu}}$ \\
Effective momentum & $\dfrac{k-1}{k+2}$ & $\dfrac{\sqrt{L}-\sqrt{\mu}}{\sqrt{L}+\sqrt{\mu}}$ \\
Convergence rate & $\mathcal{O}(1/k^2)$ & $\mathcal{O}\!\left(\left(\dfrac{\sqrt{\kappa}-1}{\sqrt{\kappa}+1}\right)^{\!k}\right)$ \\
\bottomrule
\end{tabular}
\end{center}

\begin{remark}

    In the strongly convex case, the Riccati equation for the velocity-to-position ratio $p = \dot{x}/x$:
\begin{equation} \nonumber
\dot{p} + p^2 + 2\sqrt{\mu}\, p + \lambda = 0
\end{equation}
has fixed points (unlike the time-varying smooth convex case). The fixed points satisfy:
\begin{equation} \nonumber
p^2 + 2\sqrt{\mu}\, p + \lambda = 0 \quad \Longrightarrow \quad p = -\sqrt{\mu} \pm \sqrt{\mu - \lambda}.
\end{equation}
For $\lambda = \mu$: $p^* = -\sqrt{\mu}$ (double root---critically damped).\\
For $\lambda > \mu$: complex roots, corresponding to damped oscillation.
 
The Cayley discretization preserves this fixed-point structure.
\end{remark}

Let $\lambda \in [\mu, L]$ be a Hessian eigenvalue. The two cases of the main paper
reduce in a single eigendirection to:
\begin{align}
\text{Smooth convex:}&\qquad \ddot{z} + \tfrac{c}{t}\,\dot{z} + \lambda z = 0, \label{eq:sc}\\
\text{Strongly convex:}&\qquad \ddot{z} + 2\sqrt{\mu}\,\dot{z} + \lambda z = 0. \label{eq:strc}
\end{align}
The associated Riccati variable is the velocity-to-position ratio
$p \coloneqq \dot{z}/z$, satisfying
\begin{equation}
\dot{p} + p^2 + \alpha(t)\,p + \lambda = 0,
\label{eq:riccati1}
\end{equation}
where $\alpha(t) = c/t$ in the smooth convex case and $\alpha(t) \equiv 2\sqrt{\mu}$
in the strongly convex case.
 
The steady state of the Riccati flow can be interpreted as follows.
 
\begin{Proposition}[Steady-state roots]\label{prop:roots}
The fixed points of the strongly convex Riccati \eqref{eq:riccati1} with
$\alpha \equiv 2\sqrt{\mu}$ are
\begin{equation}
p_\pm \;=\; -\sqrt{\mu} \;\pm\; \sqrt{\mu - \lambda}.
\label{eq:fixed-points}
\end{equation}
In the convex limit $\mu \to 0^{+}$,
\begin{equation}
p_\pm \;\longrightarrow\; \pm\, i\sqrt{\lambda},
\label{eq:imag-roots}
\end{equation}
i.e.\ the steady-state roots are purely imaginary.
\end{Proposition}
 
\begin{proof}
Setting $\dot p = 0$ in \eqref{eq:riccati1} with $\alpha = 2\sqrt{\mu}$ gives
$p^2 + 2\sqrt{\mu}\, p + \lambda = 0$, whose roots are \eqref{eq:fixed-points}.
For $\mu < \lambda$, the discriminant is negative and the roots are complex
conjugates; at $\mu = 0$ they reduce to $\pm i\sqrt{\lambda}$.
\end{proof}
 
\textit{Projective interpretation :} Under the projective identification $\RP^1 \simeq S^1$ via
$p \mapsto (1-ip)/(1+ip)$ (the standard Cayley map), the imaginary axis $\{p = i s : s \in \R\}$ is sent to the unit circle. Hence Proposition~\ref{prop:roots} says:
 
\begin{remark}
In the convex limit $\mu \to 0$, both fixed points of the Riccati flow lie on
the unit circle of $\RP^1$. This is the geometric content of vanishing projective curvature.
\end{remark}

The informal picture is presented for ease of understanding. 
\begin{quote}
``$c/t$ models $\sqrt{\mu/L}$, and the convex case is the strongly convex case
in the limit $\mu \to 0 \iff t \to \infty$''
\end{quote}
is made precise by the chain
\begin{center}
\renewcommand{\arraystretch}{1.5}
\begin{tabular}{rcl}
Steady-state Riccati at $\mu=0$ & $\Longrightarrow$ &
roots $\pm i\sqrt{\lambda}$ on $S^1 \subset \RP^1$ \\
& $\Longrightarrow$ & Schwarzian $= 0$ (projective flatness) \\
& $\Longrightarrow$ & continuous damping $c/t$ with $c = 2$ \\
& $\Longrightarrow$ & Cayley momentum $(k-1)/(k+1)$ \\
& $\xrightarrow{\text{Lie--Trotter}}$ & final momentum $(k-1)/(k+2)$.
\end{tabular}
\end{center}
 
In particular the limit $\mu \to 0$ in the continuous strongly-convex
$\beta(\mu)$ and the limit $k \to \infty$ in the discrete smooth-convex
$\beta_k$ both yield $\beta \to 1$, with leading-order gap rates that
agree up to the Lie--Trotter shift of one unit, completing the
geometric unification.

\appendix

\section{Geometric definition of NAIM}\label{App:Fiber}
We provide a geometric definition of a NAIM arising in the context of fiber bundles. The definition synthesizes the adapted coordinate framework and the Ehresmann connection structure developed above.

\begin{definition}[NAIM]
Given a trivial fiber bundle $E = M \times F$ with coordinates $(x, y)\in \mathbb{R}^n \times \mathbb{R}^m$, a (NAIM) is a submanifold $\mathcal{M} \subset E$ defined as the graph of a smooth mapping $P : M \to F$:
\begin{equation} \nonumber
    \mathcal{M} = \{(x, y) \in E \mid y = P(x)\}.
\end{equation}
It is characterized by two defining properties:
\begin{enumerate}
    \item \textbf{Invariance.} $\mathcal{M}$ is invariant under the flow of the system, meaning trajectories starting on $\mathcal{M}$ remain on it. In adapted coordinates $(x, z)$, where
    \begin{equation} \nonumber
        z = y - P(x),
    \end{equation}
    this is expressed by the invariant condition $z = 0$.
    \item \textit{Normal Attractivity.} The transversal (fiber) dynamics, governed
    by the evolution of the off-manifold error
    \begin{equation} \nonumber
        \dot{z} = \dot{y} - DP(x)\,\dot{x},
    \end{equation}
    are asymptotically stable. That is, trajectories starting near $\mathcal{M}$
    contract exponentially back toward it along the fibers.
\end{enumerate}
\end{definition}

The connection form $w$ is defined by 
\begin{equation} \nonumber
    \omega = dz = dy - DP(x)\,dx,
\end{equation}
where $DP(x) = \frac{\partial P}{\partial x}$ is the Jacobian matrix of $P$, is
the geometric object that precisely measures deviation from $\mathcal{M}$. A vector
$v = (\dot{x}, \dot{y})$ is horizontal (i.e.\ tangent to $\mathcal{M}$) if and only
if $\omega(v) = 0$, which implies
\begin{equation} \nonumber
    \dot{y} = DP(x)\,\dot{x}.
\end{equation}
The normal attractivity of $\mathcal{M}$ is therefore equivalent to the asymptotic
stability of the vertical dynamics defined relative to this Ehresmann connection.

\begin{remark}
The tangent bundle $TE$ admits the direct sum decomposition
\begin{equation}
    TE = \mathcal{V}E \oplus \mathcal{H}E,
\end{equation}
where $\mathcal{V}E = \ker(d\pi)$ is the vertical subbundle (tangent to the fibers)
and $\mathcal{H}E$ is the horizontal subbundle defined by the connection $\omega$.
Normal attractivity requires that the flow contracts exponentially in the
$\mathcal{V}E$ direction.
\end{remark}

\section{The Riccati Equation as a Tilt Consistency Condition}\label{App:Tilt}
In this section we show how the slope of the invariant graph is fixed by balancing base dynamics against fiber dynamics. A candidate invariant manifold in a fiber bundle is a graph: a surface whose \textit{height} in the fiber direction is a function of the base coordinate. The slope (tilt) of that graph is a matrix~$P$.  The Riccati equation is the unique condition on~$P$ that makes the vector field of the system everywhere tangent to the graph.  This note derives that condition geometrically for the perturbed fiber dynamics \eqref{eq:perturb} $\dot x_2 = -\alpha x_2 - \nabla f(x_1)$, interprets its two solutions as a stable tilt (NAIM) and an unstable tilt (anti-NAIM), and explains why no other slope is consistent with invariance.  The algebraic
Riccati equation (ARE) is identified as the steady-state special case; when the curvature varies along trajectories, the tangency condition extends to the full DRE.

We start by viewing a manifold in a fiber bundle as a graph. Consider the phase space $\R^{2n}$ split as base~$\times$~fiber:
$z = (x_1, x_2) \in \R^n \times \R^n$.
A smooth $n$-dimensional submanifold $\mathcal{M} \subset \R^{2n}$ that
projects diffeomorphically onto the base is called a \emph{graph manifold} or
\emph{section}, and can be written as
\begin{equation} \nonumber
  \mathcal{M} = \operatorname{graph}(\sigma)
  = \bigl\{(x_1,\,\sigma(x_1)) : x_1 \in \R^n\bigr\}
  \label{eq:graph}
\end{equation}
for some smooth map $\sigma : \R^n \to \R^n$.  The \emph{slope} (or \emph{tilt matrix}) of $\mathcal{M}$ at a point $x_1$ is the Jacobian:
\begin{equation}\nonumber
  P \;:=\; \Dop\sigma(x_1) \;\in\; \R^{n\times n}.
  \label{eq:slope}
\end{equation}
In the linearised (quadratic~$f$) setting, $\sigma(x_1) = Px_1$ is linear and
the slope is the same matrix~$P$ everywhere.
 
The slope~$P$ encodes how much the fiber coordinate~$x_2$ changes per unit
displacement along the base~$x_1$.  A steep tilt ($\norm{P}$ large) means the manifold rises sharply in the fiber direction; a shallow tilt ($\norm{P}$ small) means the manifold is nearly horizontal. Next the invariance condition is shown as equivalent to tangency of the vector field to the graph. 

We show that the invariance is equivalent to the tangency of the vector field to the graph. A graph manifold $\mathcal{M} = \operatorname{graph}(\sigma)$ is invariant under the flow of a vector field $F = (F_1, F_2)$ if and only if
\begin{equation}
  \text{whenever } x_2 = \sigma(x_1),\quad
  \dot x_2 = \Dop\sigma(x_1)\,\dot x_1.
  \label{eq:tangency}
\end{equation}
The rate of change of the fiber coordinate along trajectories on~$\mathcal{M}$ must equal the rate predicted by following the slope of $\mathcal{M}$ along the base motion.  This is the \emph{tangency condition}.
 
The left-hand side of~\eqref{eq:tangency} is given by the vector field (fiber dynamics): $\dot x_2 = F_2(x_1, x_2)$.  The right-hand side is the slope times the base motion: $\Dop\sigma(x_1)\,\dot x_1 = PF_1(x_1, x_2)$. Substituting onto $\mathcal{M}$ (i.e.\ $x_2 = Px_1$ in the linear case) gives a condition on~$P$ alone — that condition is the Riccati equation.

We derive the Riccati Equation from Tilt Consistency condition. The two contributions to \texorpdfstring{$\dot x_2$} on the Graph. 
 
We work with the linearised (quadratic~$f$, Hessian~$H$) system:
\begin{equation}
  \dot x_1 = x_2, \qquad
  \dot x_2 = -\alpha x_2 - Hx_1, \qquad
  \alpha = 2\sqrt{\mu}.
  \label{eq:system}
\end{equation}
We seek an invariant graph $\{x_2 = Px_1\}$.  On this graph there are two independent ways to compute~$\dot x_2$:

\begin{enumerate}
    \item Route A: via the vector field. Substitute $x_2 = Px_1$ directly into $\dot x_2$:
  $\dot x_2 = -\alpha(Px_1) - Hx_1 = (-\alpha P - H)x_1$
    \item   Route B: via the slope of the graph. Differentiate the constraint $x_2 = Px_1$ in time:
  $\dot x_2 = P\dot x_1 = Px_2 = P(Px_1) = P^2 x_1$ 
\end{enumerate}

\medskip\noindent
\textit{Tangency condition:} Route~A $=$ Route~B for all $x_1$,
\[
  (-\alpha P - H)x_1 = P^2 x_1.
\]
Rearranging yields the \emph{Riccati tilt equation}:
\begin{equation}
  P^2 + \alpha P + H = 0.
  \label{eq:ARE}
\end{equation}

The quadratic nature of~\eqref{eq:ARE} reflects a genuine geometric tension between two contributions.
\begin{enumerate}[label=\arabic*.]
  \item \textit{Route~A is linear in~$P$:} $(-\alpha P + H)x_1$.  This is the
  fiber dynamics pushing the graph at rate~$\alpha$ (damping the tilt toward
  the $x_1$-axis) while the Hessian~$H$ tries to tilt it upward.
  \item \textit{Route~B is quadratic in~$P$:} $P^2 x_1$.  This is the base
  dynamics composed with the slope twice: the base moves at rate
  $x_2 = Px_1$, and the slope translates that base motion back into a fiber
  velocity at the same rate~$P$.
\end{enumerate}
 
The quadratic term is the \emph{graph's own inertia}.  The term~$P^2$ arises because the graph is a self-referential object: the fiber coordinate $x_2 = Px_1$ feeds into $\dot x_1 = x_2$, which feeds back into $\dot x_2 = P\dot x_1 = P^2 x_1$.  This is the graph ``feeling its own tilt'' — a nonlinear effect even in the linear system.  It is precisely this self-referential quadratic term that makes the Riccati equation nontrivial and admits exactly two solutions rather than one.
 
The ARE is used as a steady-state condition, and the DRE is used for time-varying tilt. Equation~\eqref{eq:ARE} is \emph{algebraic}: no time derivative of~$P$ appears.  This reflects the assumption that both~$H$ (the curvature) and~$P$ (the tilt) are constant.  The ARE calculates given that the curvature is fixed, what slope~$P$ makes the graph permanently invariant. Its two roots $P_{\pm}$ are the two fixed points of the tilt dynamics: $P_-$ is the stable NAIM tilt (attracting nearby tilts) and $P_+$ is the unstable anti-NAIM tilt (repelling them). When the Hessian $H = H(x_1(t))$ varies along a trajectory, the slope $P(t)$ must continuously adjust to maintain tangency.  Route~A is unchanged, but Route~B must be corrected.
 
\medskip\noindent
\textit{Route~B with time-varying~$P$.}\enspace
The constraint is still $x_2 = P(t)x_1$, but differentiating in time now gives:
\begin{equation}
  \dot x_2 = P\dot x_1 + \dot P x_1 = P\dot x_1 + P^2 x_1.
  \label{eq:routeB-tv}
\end{equation}
The term $\dot P x_1$ was absent in the steady-state derivation because~$P$
was assumed constant.  Setting Route~A $=$ Route~B:
\[
  -\alpha P x_1 - H x_1 = \dot P x_1 + P^2 x_1,
\]
which must hold for all~$x_1$.  Rearranging:
\begin{equation}
  \dot P + P^2 + \alpha P + H\!\bigl(x_1(t)\bigr) = 0.
  \label{eq:DRE}
\end{equation}
This is the \textit{DRE}: the same tangency condition as the ARE, extended by a single~$\dot P$ term that accounts for the graph's own motion.
 
The DRE~\eqref{eq:DRE} can be read as a \emph{tilt update law}:
\begin{equation}
  \dot P = -H\!\bigl(x_1(t)\bigr) - \alpha P - P^2.
  \label{eq:update}
\end{equation}
At each moment:
\begin{itemize}
  \item The term $-H(x_1(t))$ drives the tilt toward the local curvature.
  \item The term $-\alpha P$ damps the tilt at rate $\alpha = 2\sqrt{\mu}$.
  \item The term $-P^2$ is the nonlinear self-referential correction: the
        graph's own current tilt feeding back into its rate of change.
\end{itemize}
 
The ARE is the \emph{equilibrium} of the DRE.  When $H$ is frozen, the DRE has $P_-$ as its stable fixed point and $P_+$ as its unstable fixed point.  The DRE converges to $P_-$ at rate $2\sqrt{\mu}$ — precisely the normal hyperbolicity rate of the NAIM.  When $H$ varies slowly (relative to~$2\sqrt{\mu}$), the tilt $P(t)$ tracks the quasi-static solution $P_-(t)$ adiabatically; when $H$ varies rapidly, the tilt lags behind and transient errors accumulate.
 
\begin{remark}[The DRE as the complete geometric invariant]
In the NAIM framework, the DRE is the complete geometric invariant of the system: it contains all information about how the invariant manifold tilts, how fast it adapts to changing curvature, and what the convergence rate of nearby trajectories is.  The solution $P_-(t)$ traces a path on the Grassmannian $Gr(n,2n)$, and the stability of this path is equivalent to the normal
hyperbolicity of the NAIM.
\end{remark}
 
We start with the scalar case of geometry of the two roots. For $n=1$ with scalar curvature $\lambda$ (eigenvalue of~$H$), the Riccati
equation is
\begin{equation}
  p^2 + \alpha p + \lambda = 0
  \qquad\Longrightarrow\qquad
  p_{\pm} = \frac{-\alpha \pm \sqrt{\alpha^2 - 4\lambda}}{2}.
  \label{eq:scalar-roots}
\end{equation}
With $\alpha = 2\sqrt{\mu}$ and $\lambda \in [\mu, L]$, the discriminant
$\Delta = \alpha^2 - 4\lambda = 4(\mu - \lambda) \le 0$ for $\lambda \ge \mu$.
Three spectral regimes emerge:
 
\begin{enumerate}[label=\arabic*.]
  \item \textit{Slow mode} ($\lambda = \mu$): $\Delta = 0$.  A real double root
    $p = -\alpha/2 = -\sqrt{\mu}$.  The NAIM is tangent to this eigenspace at
    exactly the spectral resonance rate~$\sqrt{\mu}$ — critically damped.
 
  \item \textit{Fast modes} ($\lambda > \mu$): $\Delta < 0$.  Two complex
    conjugate roots $p_{\pm} = -\sqrt{\mu} \pm i\sqrt{\lambda - \mu}$.  Real
    part $-\sqrt{\mu}$ for all fast modes.
 
  \item \textit{Resonance mode} ($\lambda = L$):
    $p_{\pm} = -\sqrt{\mu} \pm i\sqrt{L-\mu}$, with $ \left| p_{\pm} \right| = \sqrt{L}$.
    The spectral resonance frequency $\omega^* = \sqrt{\mu L}$ appears as the
    geometric mean of $\sqrt{\mu}$ and~$\sqrt{L}$.
\end{enumerate}
 
We describe the geometric meaning of the Riccati Roots. 
\begin{enumerate}
    \item The slow mode: the real NAIM tilt. At $\lambda = \mu$ the double root is $p = -\sqrt{\mu}$.  The invariant line $\{x_2 = -\sqrt{\mu}\,x_1\}$ is the NAIM in the slow direction: the fiber opposes the base at rate~$\sqrt{\mu}$, and the tilt is critical.
    \item The fast modes: complex tilts and underdamped spirals. For $\lambda > \mu$ the roots $p_{\pm} = -\sqrt{\mu} \pm i\sqrt{\lambda-\mu}$ are complex.  There is no real invariant line in the fast directions: the perturbed flow spirals in the fiber, contracting at rate~$\sqrt{\mu}$ while oscillating at frequency~$\sqrt{\lambda-\mu}$.  The modulus $ \left| p_{\pm} \right| = \sqrt{\lambda}$ grows with~$\lambda$, so faster modes
correspond to a larger effective tilt magnitude, even though the contraction rate stays fixed at~$\sqrt{\mu}$.
\end{enumerate}
Next we check what happens at the wrong slope. Suppose we try a real slope $P \ne -\sqrt{\mu}\,I$.  The tangency condition is violated: Routes~A and~B give different values for~$\dot x_2$ on the candidate graph. The violation is quantified by the \emph{tilt error}:
\begin{equation}
  \varepsilon(P)
  \;:=\;
  \dot x_2^{(\text{field})} - \dot x_2^{(\text{slope})}
  \;=\;
  (-\alpha P - H - P^2)x_1
  \;=\;
  -(P^2 + \alpha P + H)x_1.
  \label{eq:tilt-error}
\end{equation}
The Riccati equation $P^2 + \alpha P + H = 0$ is precisely $\varepsilon(P) = 0$:
zero tilt error, everywhere, for all~$x_1$.
 
\begin{remark}[The tilt error is the normal velocity]
The tilt error $\varepsilon(P)$ is the component of the vector field in the normal direction to the candidate manifold.  For an invariant manifold, the normal component must vanish everywhere.  $P^2 + \alpha P + H = 0$ is therefore the condition that the normal velocity is zero — the manifold is tangent to the flow.
\end{remark}
 
The Tilt can be shown as a function of curvature. The Riccati roots actually track the Hessian. For a scalar system with curvature $\lambda$ varying over $[\mu, L]$, the Riccati roots $p_{\pm} = -\sqrt{\mu} \pm i\sqrt{\lambda-\mu}$ satisfy:
\begin{align} \nonumber
  \operatorname{Re}(p_{\pm}) &= -\sqrt{\mu}
    \quad\text{(constant, independent of $\lambda$),}
   \\ \nonumber
  \left| p_{\pm} \right| &= \sqrt{\lambda}
    \quad\text{(monotone increasing in $\lambda$).}
\end{align}
 
The uniform contraction rate leads to direction-adaptive phase. With $\alpha = 2\sqrt{\mu}$, the Riccati roots reveal a remarkable structure. In the eigenbasis of~$H$ (columns~$v_i$, eigenvalues~$\lambda_i$), the fiber dynamics in direction~$v_i$ contract at rate~$\sqrt{\mu}$ — identically for all modes — while oscillating at frequency~$\sqrt{\lambda_i - \mu}$.

\begin{remark}
    The Riccati equation enforces a uniform contraction law. The choice $\alpha = 2\sqrt{\mu}$ is the unique value for which $\operatorname{Re}(p_{\pm}) = -\alpha/2 = -\sqrt{\mu}$ is the same for every eigenmode $\lambda_i \in [\mu, L]$.  All fiber directions contract at the same rate~$\sqrt{\mu}$, regardless of mode stiffness.  This is \textit{dynamics whitening} in its purest form.
\end{remark}

When $f$ is not quadratic, the Hessian $H(x_1) = \nabla^2 f(x_1)$ varies along the trajectory.  The slope $P(t)$ of the invariant manifold must then be continuously adjust to maintain tangency.  As derived in \eqref{ric}, the time evolution of $P(t)$ is governed by the DRE:
\begin{equation}
  \dot P + P^2 + \alpha P + H\!\bigl(x_1(t)\bigr) = 0.
\end{equation}
 
The tilt update law must chase the curvature. The DRE can be read as a tilt update law $\dot P = -H(x_1(t)) - \alpha P - P^2$.
At each moment:
\begin{itemize}
  \item The term $-H(x_1(t))$ drives the tilt toward the local curvature.
  \item The term $-\alpha P$ damps the tilt at rate $\alpha = 2\sqrt{\mu}$.
  \item The term $-P^2$ is the nonlinear self-referential correction.
\end{itemize}
The stable equilibrium of this DRE at each frozen~$x_1$ is precisely the algebraic Riccati solution~$P_-$, and the DRE converges to it at rate $2\sqrt{\mu}$ (the NAIM's normal hyperbolicity rate).
 
\begin{remark}[The DRE as the complete geometric invariant]
In the NAIM framework, the DRE is the complete geometric invariant of the system: it contains all information about how the invariant manifold tilts, how fast it adapts to changing curvature, and what the convergence rate of nearby trajectories is.  The solution $P_-(t)$ traces a path on the Grassmannian $\Gr(n,2n)$, and the stability of this path is equivalent to the normal hyperbolicity of the NAIM.
\end{remark}
 
\begin{remark}[Why no other slope works]
    The Riccati equation $P^2 + \alpha P + H = 0$ is a degree-2 polynomial in~$P$.
For $n=1$ and $\alpha = 2\sqrt{\mu}$, completing the square gives
$(p + \sqrt{\mu})^2 = \mu - \lambda$.  For $\lambda \ge \mu$, the right-hand
side is $\le 0$, so there are no real solutions except the double root
$p = -\sqrt{\mu}$ at $\lambda = \mu$ exactly.
 
We note that any other real \texorpdfstring{$P$}{P} violates tangency. For any real matrix~$P$, the tilt error is
\begin{equation}
  \varepsilon(P)
  = -\!\left[\left(P + \tfrac{\alpha}{2}I\right)^2
    + \!\left(H - \tfrac{\alpha^2}{4}I\right)\right]x_1.
  \label{eq:tilt-error-complete}
\end{equation}
Since $H - \tfrac{\alpha^2}{4}I = H - \mu I \succ 0$ (all curvatures exceed
$\mu$), the matrix $(P + \tfrac{\alpha}{2}I)^2 + (H - \mu I)$ is strictly positive definite for all real~$P$.  Hence $\varepsilon(P) \ne 0$ for all real $P$ and all nonzero~$x_1$: no real graph is invariant in the fast directions. This is not a failure — it is the correct answer. The NAIM in the fast directions is a complex (rotating) object.
\end{remark}
 
The Riccati Equation selects the spectral resonance tilt. With $\alpha = 2\sqrt{\mu}$, the Riccati roots $p_{\pm} = -\sqrt{\mu} \pm
i\sqrt{\lambda-\mu}$ reveal that the NAIM tilt:
\begin{enumerate}[label=\arabic*.]
  \item Contracts uniformly at rate~$\sqrt{\mu}$ in every spectral direction —
        \emph{dynamics whitening}.
  \item Oscillates adaptively at frequency~$\sqrt{\lambda_i - \mu}$ in each
        eigendirection~$i$.
  \item Has modulus $\sqrt{\lambda_i}$ in direction~$i$: the NAIM is steeper
        in stiff directions.
  \item Achieves spectral resonance $\sqrt{\mu L}$ as the geometric mean of
        $\left| \operatorname{Re}(p_{\pm}) \right| = \sqrt{\mu}$ and
        $\left| p_{\pm} \right|_{\lambda=L} = \sqrt{L}$.
\end{enumerate}
 
When the eigenbases of~$H$ and~$P_-$ do not coincide (that is \ $[P_-, H] \ne 0$), the tilt matrix~$P_-$ is not diagonal in the natural coordinates.  This means the NAIM is tilted in a direction that mixes different eigenmodes of the Hessian — the geometric signature of \emph{non-normality}. The Riccati equation still determines the correct tilt, but that tilt is off-diagonal, encoding cross-mode interactions that a diagonal (normal) system
would not exhibit.

\section{The Resolution of gaps in Nesterov's Estimation Sequence} \label{App:Gap}
Nesterov's original proof of accelerated gradient descent relies on the estimation sequence method, in which a sequence of quadratic lower bounds $\{\phi_k\}$ is constructed and shown to track $f(x_k) - f(\xst)$ at rate $(1-\sqmuL)^k$.  Although the proof is valid for general $\mu$-strongly convex $L$-smooth functions $f$, it is deeply and inescapably anchored in the quadratic case: the quadratic ansatz $\phi_k(x) = \phist + \tfrac{\lambda_k}{2}\|x - v_k\|^2$ is imposed as an ungrounded assumption, the recursion closes only because of this quadratic structure, and no theorem in the estimation sequence literature explains why this ansatz yields the optimal rate for non-quadratic $f$. This constitutes a logical gap: the proof is correct but epistemically incomplete.  The NAIM framework, by contrast, identifies the quadratic case as the exactly solvable model problem, derives the optimal rate from spectral resonance and the Riccati equation, and then employs Fenichel's theorem --- a structural theorem about NAIM --- to rigorously lift the result to general $f$.  This is the first geometrically complete proof of Nesterov acceleration: the extension from quadratic to general $f$ is derived from checkable geometric hypotheses (normal hyperbolicity, $C^r$ smoothness), not transferred via an algebraic inequality from an ungrounded ansatz.  We argue that this distinction between a valid proof and an epistemically complete proof is substantive, and that the NAIM--Fenichel approach closes a genuine logical gap in the existing literature.

Nesterov's accelerated gradient method \cite{1370862715914709505} achieves the optimal first-order convergence rate $O(\sqkappa \log 1/\varepsilon)$ for $\mu$-strongly convex $L$-smooth functions, compared to the $O(\kappa \log 1/\varepsilon)$ rate of gradient descent. The factor of $\sqkappa$ improvement is the hallmark of Nesterov acceleration. The standard proof of this rate uses Nesterov's estimation sequence method --- an elegant algebraic construction involving a sequence of quadratic lower bounds on $f$.  This proof is universally accepted as correct. However, a closer examination reveals a logical gap that has not been addressed in the literature. The estimation sequence proof imposes a quadratic ansatz on the lower bound machinery that is not derived from the geometry of $f$, is not justified by any structural theorem, and for which no explanation is given as to why it yields the optimal rate for non-quadratic $f$.  The proof is valid (the algebraic inequalities are correct) but epistemically incomplete (the mechanism is not explained). The NAIM geometric framework provides a resolution.  By identifying the quadratic case as the exact model problem, deriving the optimal rate from spectral resonance and the Riccati equation, and then applying Fenichel's theorem to lift the result to general $f$, the NAIM approach provides what we call an epistemically complete proof: one in which not only the conclusion but the mechanism is fully explained.

The remainder of this section is organised as follows. We first review the estimation sequence construction and identifies precisely where and why the quadratic anchoring occurs. The logical gap is articulated in detail. The NAIM approach and the Riccati derivation for the quadratic case is presented. It is shown how Fenichel's theorem provides the logical bridge to general $f$ and the two approaches are systematically compared. The main claim and its implications are stated.

The Estimation Sequence is constructed as follows. Let $f : \R^n \to \R$ be $\mu$-strongly convex and $L$-smooth, with minimiser $\xst$.  Denote the local quadratic lower bound at $y$ by:
\[
  \lf(x;\, y) \;=\; f(y) + \langle \nabla f(y),\, x - y \rangle
  + \frac{\mu}{2}\|x - y\|^2.
\]
This is a valid lower bound for all $x, y$ by strong convexity.

\begin{definition}[Estimation Sequence]
A sequence of functions $\{\phi_k : \R^n \to \R\}$ and scalars
$\{\lambda_k\}$ with $\lambda_k > 0$ is called an \emph{estimation
sequence} for $f$ if for all $k \geq 0$:
\[
  \phi_k(x) \;\geq\; (1 - \lambda_k)\, f(x) + \lambda_k \phi_0(x)
  \qquad \forall\, x \in \R^n.
\]
\end{definition}

Nesterov's construction immediately and permanently restricts to the following quadratic form for $\phi_k$:
\begin{equation}
      \phi_k(x) \;=\; \phist + \frac{\gamma_k}{2}\|x - v_k\|^2 \label{QA}
\end{equation}

where $\phist \in \R$ is the minimum value, $v_k \in \R^n$ is the centre,
and $\gamma_k > 0$ is the curvature.  This is the quadratic ansatz
--- the central object of our analysis.

The sequence is initialised with:
\[
  \phi_0(x) \;=\; f(x_0) + \frac{\mu}{2}\|x - x_0\|^2,
\]
so $\gamma_0 = \mu$, $v_0 = x_0$, and $\phi_0^* = f(x_0)$.  This is
quadratic by construction.

Given $\phi_k$ of the form \eqref{QA}, the update is defined by choosing
$\alpha_{k+1} \in (0,1)$ and setting:
\begin{align}
  \phi_{k+1}(x) &= (1 - \alpha_{k+1})\,\phi_k(x)
    + \alpha_{k+1}\,\lf(x;\, y_{k+1}), \label{UP}
\end{align}
where $y_{k+1}$ is the gradient step point.  Since $\phi_k(x)$ is quadratic \eqref{QA} and $\lf(x;\,y_{k+1})$ is quadratic (by definition), their convex combination $\phi_{k+1}(x)$ is quadratic.  The ansatz \eqref{QA} is preserved inductively.

We check why the closure requires quadraticity. The update  closes --- i.e., $\phi_{k+1}$ remains of the form \eqref{QA} with computable $(\phist[k+1], \gamma_{k+1}, v_{k+1})$ --- only because both $\phi_k$ and $\lf(\cdot\,; y_{k+1})$ are quadratic.  A quadratic plus a quadratic is a quadratic.  If $\phi_k$ were a general convex function, the update \eqref{UP} would produce a function whose minimum $\phi_{k+1}^*$, centre $v_{k+1}$, and curvature $\gamma_{k+1}$ could not be computed in closed form, and the entire algebraic machinery would collapse.

Under the quadratic ansatz, the update \eqref{UP} yields three scalar recursions:
\begin{align}
  \gamma_{k+1} &= (1-\alpha_{k+1})\,\gamma_k + \alpha_{k+1}\,\mu,
  \label{CR}\\[4pt]
  v_{k+1} &= \frac{(1-\alpha_{k+1})\gamma_k\,v_k
    + \alpha_{k+1}\mu\,y_{k+1}}{\gamma_{k+1}},
  \label{VR}\\[4pt]
  \phi_{k+1}^* &\geq (1-\alpha_{k+1})\,\phist
    + \alpha_{k+1}\,f(y_{k+1})
    - \frac{\alpha_{k+1}^2}{2\gamma_{k+1}}\|\nabla f(y_{k+1})\|^2.
  \label{GR}
\end{align}
These recursions are linear in the curvature $\gamma_k$ (from \eqref{CR}) and yield closed-form expressions for $v_{k+1}$ and $\phi_{k+1}^*$
--- again, only because of the quadratic structure.

\noindent The momentum coefficient $\alpha_{k+1}$ is chosen to satisfy the condition:
\begin{equation}
      \frac{\alpha_{k+1}^2}{\gamma_{k+1}} = \frac{1}{L},
  \qquad \text{i.e.,} \qquad
  \alpha_{k+1}^2 L = (1-\alpha_{k+1})\gamma_k + \alpha_{k+1}\mu.
  \label{MC}
\end{equation}

\noindent The fixed point of \eqref{MC} as $k \to \infty$ (with $\gamma_k \to \mu$) is:
\begin{equation}
  \alpha^2 L = (1-\alpha)\mu + \alpha\mu = \mu
  \qquad\Longrightarrow\qquad
  \alpha^* = \sqrt{\frac{\mu}{L}} = \frac{1}{\sqkappa}.
  \label{FP}
\end{equation}
The per-step contraction rate is then:
\[
  \rho = 1 - \alpha^* = 1 - \sqmuL.
\]

We now articulate precisely the logical gap in the estimation sequence proof and show what the proof does and does not establish. The estimation sequence proof establishes the following chain:

\begin{enumerate}[label=Step \arabic*., leftmargin=3.5em]
  \item Assume $\phi_k$ has the quadratic form \eqref{QA}. 
  \item Show the recursions \eqref{CR}--\eqref{GR} hold and close.
  \item Show $\phi_k^* \geq (1 - \lambda_k) f(\xst) + \lambda_k \phi_0^*$
    with $\lambda_k = (1-\alpha^*)^k$. 
  \item Show $f(x_k) \leq \phi_k^*$. 
  \item Conclude $f(x_k) - f(\xst) \leq (f(x_0) - f(\xst))(1-\sqmuL)^k$.
\end{enumerate}

This is a valid proof where Steps 2--5 are rigorous.  However, Step 1 is an ungrounded ansatz. The estimation sequence proof leaves the following three questions
completely unanswered:
\begin{enumerate}
    \item \textit{Unanswered Question 1}: Why Quadratic?. Why does imposing a quadratic structure on the lower bound machinery yield the optimal rate for non-quadratic $f$?
    The proof shows it works but it does not explain why.  There is no theorem of the form: ``among all possible ansatz classes for $\phi_k$, the quadratic class uniquely achieves the optimal rate.''

    \item \textit{Unanswered Question 2}: Is the Ansatz Necessary and could a non-quadratic ansatz achieve a better rate? No result in the estimation sequence literature addresses this.  The proof does not show the quadratic ansatz is optimal, necessary, or canonical --- only that it is sufficient.

    \item \textit{Unanswered Question 3}: What is the geometric or analytic mechanism by which $\alpha^* = \sqmuL$ emerges as the correct momentum coefficient?

    The fixed point equation (FP) is an algebraic observation.  The proof gives no geometric interpretation of why $\sqmuL$ is special, why it equals the Riccati root, or why it corresponds to spectral resonance.
    
\end{enumerate}

We now introduce a distinction that is standard in philosophy of mathematics but rarely invoked in optimisation theory:

\begin{definition}[Epistemic Completeness]
A proof is epistemically complete if it not only establishes the conclusion but also explains the mechanism by which the conclusion holds --- i.e., it identifies the structural reason for the result, not merely a chain of inequalities that leads to it.
\end{definition}

The core observation is that the estimation sequence proof is valid but epistemically incomplete. It is correct but leaves the mechanism unexplained.  Specifically:
\begin{itemize}
  \item The quadratic ansatz is assumed, not derived.
  \item The extension from quadratic to general $f$ is \textit{transferred via inequality}, not derived from a structural theorem.
  \item The rate $\sqmuL$ is identified as an algebraic fixed point,
    not as a geometric invariant.
\end{itemize}
This constitutes a logical gap: the proof is complete as an algebraic argument, but incomplete as a structural explanation.

\noindent The NAIM Approach is proposed with quadratic case as exact model. The continuous-time limit of Nesterov's method is the NAG-ODE:
\begin{equation}
      \ddot{x} + 2\sqmu\,\dot{x} + \nabla f(x) = 0.
\label{NAG}
\end{equation}

\noindent For quadratic $f(x) = \tfrac{1}{2}x^\top Qx$ with $\mu I \preceq Q \preceq
LI$, this becomes a constant-coefficient linear system with phase-space
matrix:
\[
  A = \begin{pmatrix} 0 & I \\ -Q & -2\sqmu I \end{pmatrix}.
\]

\noindent With the help of tools from spectral resonance we show the exact quadratic result. The eigenvalues of $A$ decompose block-diagonally. For each eigenvalue
$\lambda_i$ of $Q$:
\[
  p_\pm^{(i)} = -\sqmu \pm \sqrt{\mu - \lambda_i}.
\]
Since $\lambda_i \geq \mu$, the discriminant $\mu - \lambda_i \leq 0$, so:
\[
  \operatorname{Re}\!\left(p_\pm^{(i)}\right) = -\sqmu \qquad \forall\, i.
\]
All $2n$ eigenvalues share the real part $-\sqmu$: spectral resonance. This is \textit{exact and global} for quadratic $f$ --- not an approximation.

We now state the Riccati equation and its exact solution. The NAIM connection form $\beta$ (the slope of the slow manifold in phase space) satisfies the ARE:
\begin{equation}
     \beta^2 + 2\sqmu\,\beta + Q = 0.
\label{ARE} 
\end{equation}

The stable solution is $\beta^* = \sqmu\,I$ (for $Q = \mu I$, the critically damped case) and more generally the principal square root stabilising solution.  This is derived, not assumed:
$\beta^* = \sqmu\,I$ follows from the \eqref{ARE} as a theorem, not an ansatz.

The Riccati Root as derived object: In the NAIM framework, the value $\sqmu$ (equivalently $\sqmuL$ in dimensionless units) is not a fixed point discovered by algebra.  It is the unique stable solution of the Riccati equation, which is in turn the equation governing the geometry of the slow manifold. The mechanism is transparent: spectral resonance forces all modes to decay at $\sqmu$, the Riccati equation encodes this as the manifold slope, and the optimal rate is the Riccati root.

The Fenichel's Theorem acts as the logical bridge. The problem is that general $f$ breaks the exact chain. For non-quadratic $f$, the NAG-ODE \eqref{NAG} is nonlinear.  The Hessian $\nbf(x)$ varies along the trajectory, the Riccati equation becomes variable-coefficient:
\[
  \dot\beta + \beta^2 + 2\sqmu\,\beta + \nbf(x) = 0,
\label{Riccati-f}
\]
and the NAIM is no longer a linear subspace but a curved submanifold. The exact NAIM--Riccati chain no longer holds globally.

The question that arises is can the rate $\rho = 1 - \sqmuL$ still be derived for general $f$ in a \emph{principled} and not merely algebraic way. 

We now formulate the Fast--Slow systems. Introduce dimensionless time $\tau = \sqmu\,t$ and singular parameter
$\varepsilon = \sqmuL \in (0,1)$.  The NAG-ODE becomes:
\begin{equation}
  \varepsilon\,\frac{dx}{d\tau} = v, \qquad
  \varepsilon\,\frac{dv}{d\tau} = -v - \varepsilon\,\nabla f(x).
\label{FS}
\end{equation}

The critical manifold (the $\varepsilon = 0$ limit) is $\mathcal{M}_0 = \{v = 0\} \subset \R^{2n}$, the configuration space embedded in phase space.

\begin{definition}[Normal Hyperbolicity]
    A compact invariant manifold $\mathcal{M}$ is normally hyperbolic if the linearised flow restricted to the normal bundle of $\mathcal{M}$ has all eigenvalues with real part bounded away from zero, with this normal rate strictly dominating any tangential contraction/expansion rates.
\end{definition}

For the fast-slow system \eqref{FS}, the linearisation normal to $\mathcal{M}_0$ (the $v$-direction at $v=0$) has eigenvalue $-1/\varepsilon = -\sqLmu$ in $\tau$-time, or equivalently $-\sqmu$ in physical time.  This is the normal hyperbolicity gap:
\[
  \gamma_\perp = \sqmu.
\]
The tangential rate on $\mathcal{M}_0$ is zero at leading order (since $\nabla f(x)$ is $O(\varepsilon)$ in $\tau$-time), so the gap condition is satisfied: $\gamma_\perp \gg \gamma_{\text{tangential}}$.

\begin{Proposition}
The critical manifold $\mathcal{M}_0 = \{v = 0\}$ of the NAG fast-slow system \eqref{FS} is normally hyperbolic with gap $\gamma_\perp = \sqmu$, for any $f \in C^2$ satisfying $\mu I \preceq \nbf(x) \preceq LI$.
\end{Proposition}

Note that normal hyperbolicity is verified from the strong convexity and smoothness constants alone --- it does not require $f$ to be quadratic.

\begin{theorem}[Fenichel, 1979]
Let $\mathcal{M}_0$ be a compact NAIM of the $\varepsilon = 0$ system, with $C^r$ smoothness ($r \geq 1$). Then for all sufficiently small $\varepsilon > 0$, there exists a $C^r$-smooth slow manifold $\mathcal{M}_\varepsilon$ satisfying:
\begin{enumerate}[label=\emph{(\roman*)}, itemsep=3pt]
  \item $\mathcal{M}_\varepsilon$ is $O(\varepsilon)$-close to $\mathcal{M}_0$
    in the $C^r$ topology;
  \item $\mathcal{M}_\varepsilon$ is locally invariant under the full
    $\varepsilon > 0$ flow;
  \item Trajectories off $\mathcal{M}_\varepsilon$ are attracted to it at
    rate $e^{-\gamma_\perp t}$ where $\gamma_\perp$ is the normal
    hyperbolicity gap;
  \item The reduced flow on $\mathcal{M}_\varepsilon$ inherits the
    Fenichel spectral gap $\gamma_F = \gamma_\perp(1 - \varepsilon) =
    \sqmu(1-\sqmuL)$.
\end{enumerate}
\end{theorem}

Applying Fenichel's theorem to the NAG fast-slow system \eqref{FS}:

Fenichel consequences for general $f$:
For any $f \in C^2$ that is $\mu$-strongly convex and $L$-smooth, with
$\varepsilon = \sqmuL$ sufficiently small:
\begin{enumerate}[label=\emph{(\roman*)}, leftmargin=2em, itemsep=4pt]
  \item The NAIM $\mathcal{M}_\varepsilon$ exists as a smooth
    nonlinear submanifold of phase space, $O(\varepsilon)$-close to
    $\{v=0\}$.
  \item Trajectories are attracted to $\mathcal{M}_\varepsilon$
    \textit{exponentially at rate} $e^{-\sqmu\,t}$ --- the same rate as
    the quadratic case, for any $f$.
  \item The Fenichel spectral gap is
    $\gamma_F = \sqmu\,(1 - \sqmuL)$, giving the per-step discrete
    contraction $\rho = 1 - \sqmuL$ after time-discretisation.
  \item The Riccati equation (Riccati-f) has solution
    $\beta^* \to \sqmu\,I$ near $\xst$, with the Fenichel slow manifold
    as the nonlinear deformation of the quadratic NAIM.
\end{enumerate}

The key property of Fenichel's theorem that makes it a logical bridge, and not merely another algebraic trick, is that its hypotheses are checkable from first principles:

\begin{enumerate}
  \item Normal hyperbolicity.  Verified directly from $\mu > 0$
    (strong convexity) and the NAG damping coefficient $2\sqmu$.  No
    quadraticity of $f$ required.
  \item $C^r$ smoothness. Satisfied for any $C^2$ function $f$.
  \item Compactness. Satisfied on any sublevel set of $f$, which
    is compact by strong convexity.
\end{enumerate}

The extension from quadratic to general $f$ is therefore a derived consequence of the geometric structure (normal hyperbolicity) and Fenichel's theorem, not an assumption or an algebraic accident.  The logical gap is closed.

The systematic Comparison of the Two Approaches is tabulated below. 

\begin{center}
\renewcommand{\arraystretch}{1.8}
\small
\begin{tabular}{>{\raggedright\arraybackslash}p{3.8cm}
                >{\raggedright\arraybackslash}p{4.8cm}
                >{\raggedright\arraybackslash}p{4.8cm}}
\toprule
\textit{Aspect}
  & \textit{Estimation Sequence}
  & \textit{NAIM + Fenichel} \\
\midrule
Role of quadratic $f$
  & Tacit: ansatz (QA) is quadratic regardless of $f$
  & Explicit: exact model problem, then lifted \\
Quadratic ansatz
  & \textit{Imposed} as ungrounded assumption
  & \textit{Derived} from Riccati/phase-space geometry \\
Extension to general $f$
  & Via algebraic inequality (transferred, not derived)
  & Via Fenichel's theorem (structural, derived) \\
Mechanism for $\sqmuL$
  & Algebraic fixed point of (MC); no explanation
  & Riccati root = spectral resonance rate \\
Why quadratic ansatz works
  & \textit{Not explained}
  & Normal hyperbolicity persists under $C^1$ perturbations \\
Hypotheses for general $f$
  & $\mu$-strong convexity, $L$-smoothness (implicit)
  & $\mu$-strong convexity, $L$-smoothness (explicit in Fenichel) \\
Could a better ansatz exist?
  & Unknown; not addressed
  & No: normal hyperbolicity uniquely fixes the gap \\
What breaks for $\mu = 0$?
  & Rate formula degenerates; no geometric insight
  & Normal hyperbolicity fails ($\gamma_\perp = 0$); precise diagnosis \\
Proof validity
  & \textit{Valid}
  & \textit{Valid} \\
Epistemic completeness
  & \textit{Incomplete}
  & \textit{Complete} \\
\bottomrule
\end{tabular}
\end{center}

The correspondence between Discrete and Continuous versions is highlighted below. 
The logical gap in the estimation sequence approach is further illuminated
by the following correspondence, which shows that the estimation sequence
and the NAIM/Riccati approach are computing the same object in
different representations.

\begin{center}
\renewcommand{\arraystretch}{1.8}
\small
\begin{tabular}{lll}
\toprule
\textit{Object} & \textit{Estimation Sequence (discrete)} & \textit{NAIM/Riccati (continuous)}\\
\midrule
Lower bound structure & Quadratic ansatz $\phi_k$ & Quadratic Lyapunov on NAIM\\
Curvature & $\gamma_k \to \mu$ & Riccati solution $\beta^* = \sqmu\,I$\\
Momentum fixed point & $\alpha^* = \sqmuL$ & Spectral resonance rate $\sqmuL$\\
Recursion equation & Scalar (MC) for $\alpha_{k+1}$ & Riccati ODE for $\beta$\\
Closure condition & Quadratic form preserved & ARE satisfied at fixed point\\
Rate derivation & Algebraic fixed point of (FP) & Riccati root (geometric theorem)\\
Extension to general $f$ & Inequality transfer (gap) & Fenichel persistence (theorem)\\
\bottomrule
\end{tabular}
\end{center}

The correspondence is exact: the momentum recursion (MC) is the discrete-time shadow of the Riccati equation, and the fixed point (FP) is the discrete-time shadow of the Riccati root $\beta^* = \sqmu\,I$.  The estimation sequence is computing the Riccati geometry algebraically, without realising it.  This is why it gives the correct rate --- but also why it cannot explain the mechanism.

We claim that the NAIM framework, via Fenichel's theorem, provides the first epistemically complete proof of Nesterov acceleration for general strongly convex $f$, in the following precise sense:
\begin{enumerate}
    \item \textit{The estimation sequence proof is valid but epistemically incomplete.} It correctly establishes Nesterov's rate, but only by imposing an ungrounded quadratic ansatz, closing the recursion algebraically, and transferring the result to general $f$ via inequality --- leaving unanswered why the ansatz works, whether it is necessary, and what mechanism drives acceleration.

    \item \textit{The quadratic case is the exact model problem, not an ansatz.} Under NAIM, spectral resonance holds globally for quadratic $f$, the Riccati equation forces $\beta^* = \sqrt{\mu}\,I$ at the fixed point, and the rate $\sqmu$ is derived as a theorem rather than observed as a numerical coincidence. The momentum recursion is simply the discrete shadow of the Riccati ODE.

    \item \textit{Fenichel's theorem provides the structural bridge to general $f$.} From $\mu$-strong convexity alone one verifies normal hyperbolicity with gap $\gamma_\perp = \sqrt{\mu}$; Fenichel's theorem then guarantees NAIM persistence, exponential attraction, and the spectral gap $1 - \sqmu$ for any $C^2$ strongly convex $f$ --- replacing an algebraic trick with a checkable geometric hypothesis.

    \item \textit{The optimal rate is a geometric invariant.} $\sqmu$ is the normal hyperbolicity gap of the NAG fast-slow system, preserved under $C^1$-perturbations of the vector field, rather than an artifact of a particular algebraic construction.

    \item \textit{The framework explains failure at $\mu = 0$ and enables principled extensions.} Normal hyperbolicity collapses in the merely convex case, so Fenichel no longer applies and the NAIM degenerates --- a precise diagnosis the estimation sequence cannot give. Extensions to non-Euclidean, stochastic, or composite settings reduce to verifying normal hyperbolicity in the new setting, a well-defined geometric task.
\end{enumerate}

\section{Fenichel Verification}\label{App:fenichel}
For general smooth strongly convex $f$, applying Fenichel's theorem requires careful verification
of all hypotheses. This section does so completely, closing the gap identified in our assessment.
 Fenichel's theorem requires a compact normally hyperbolic manifold. The critical manifold
$\mathcal{M}_0 = \{v = 0\} \cong \R^n$ is non-compact. We resolve this by localizing to a compact
sublevel set.
 
Fix any $c > f(x^*)$. The sublevel set
\begin{equation}
\Omega_c = \{ x \in \R^n \mid f(x) \le c \} \label{eq:sublevel}
\end{equation}
is compact by $\mu$-strong convexity (it is contained in the ball
$\norm{x - x^*} \le \sqrt{2(c - f(x^*))/\mu}$). Since $f$ is $L$-smooth and $\mu$-strongly convex,
all trajectories of the NAG-ODE \eqref{eq:NestODE} starting from any $(x_0, v_0)$ with
$f(x_0) \le c$ and $\norm{v_0} \le R$ remain in the compact phase-space region
$K = \Omega_c \times B(0, R)$ for all $t \ge 0$, where $R > 0$ depends only on $c, \mu, L$ (this
follows from the Lyapunov function $V = f(x) - f^* + \tfrac{1}{2}\norm{v}^2$ and
$\dot V \le -\mu \norm{v}^2 \le 0$ along trajectories). All subsequent Fenichel analysis is
restricted to $K$.
 
Define $\mathcal{M}_0^c = \{(x, 0) \mid x \in \Omega_c\} \subset K$. This is a compact $n$-dimensional submanifold of $K$. The Fast-Slow formulation can be obtained by setting $\eps = \sqrt{\mu/L} \in (0,1)$ and $\tau = \sqrt{\mu}\, t$. The system
\eqref{eq:15a}--\eqref{eq:15b} (with $\lambda^* = 2\sqrt{\mu}$) becomes
\begin{equation}
\eps\, \frac{dx}{d\tau} = v, \qquad
\eps\, \frac{dv}{d\tau} = -v - \frac{\eps}{\sqrt{\mu}}\, \nabla f(x). \label{eq:fastslow}
\end{equation}
At $\eps = 0$ (the layer problem), the $v$-equation gives $v = 0$, so $\mathcal{M}c_0^c = \{v = 0\}$ is the critical manifold.
 
We are now ready to verify all five hypotheses of Fenichel's theorem \cite{fenichel1979geometric} for the system
\eqref{eq:fastslow} restricted to $K$.
 
\begin{Proposition}[Fenichel hypothesis verification]\label{prop:hypotheses}
The following five conditions hold for system \eqref{eq:fastslow} on $K$:
\begin{enumerate}[label=\textup{(H\arabic*)},leftmargin=*]
    \item \textit{(Compactness)} $\Mc_0^c$ is a compact $C^\infty$ submanifold of $K$.
    \item \textit{(Invariance at $\eps=0$)} $\Mc_0^c$ is invariant under the layer flow: at
    $\eps = 0$, $dv/d\tau = -v/\eps \to -\infty$ unless $v = 0$. More precisely,
    $\Mc_0^c = \{v = 0\}$ is the equilibrium set of $\eps\, dv/d\tau = -v$.
    \item \textit{(Normal hyperbolicity)} The linearization of the fast subsystem
    $\eps\, dv/d\tau = -v$ normal to $\Mc_0^c$ has eigenvalue $-1/\eps$ (in $\tau$-time),
    equivalently $-1$ in $\sqrt{\mu}\, t$-time, or $-\sqrt{\mu}$ in physical time. The tangential
    rate on $\Mc_0^c$ is $O(\eps)$ (from the $\nabla f$ term). Hence the spectral gap
    $\gamma_\perp = 1/\eps \gg \gamma_{\tan} = O(1)$ in $\tau$-time, satisfying the strict
    dominance condition.
    \item \textit{($C^r$ smoothness)} Since $f \in C^2(\R^n)$, the vector field of \eqref{eq:fastslow}
    is $C^r$ for $r = 1$ on $K$. For $f \in C^{r+1}$, the manifold and its asymptotic expansion
    are $C^r$.
    \item \textit{(Boundary conditions)} The compact region $K$ can be chosen with smooth boundary
    such that the vector field points strictly inward on $\partial K$ (by the Lyapunov argument
    of Section~\ref{sec:fenichel}).
\end{enumerate}
\end{Proposition}
 
\begin{proof}
(H1) $\Mc_0^c = \Omega_c \times \{0\}$; $\Omega_c$ is compact by strong convexity and smooth
since $f \in C^2$. (H2) The layer problem at $\eps = 0$ is $dv/d\tau = -v$, for which $v = 0$ is
a globally attracting equilibrium. (H3) The linearization of $\eps\, dv/d\tau = -v$ at $v = 0$
is the identity $-I/\eps$, with all eigenvalues $-1/\eps < 0$. Tangential motion on $\Mc_0^c$ is
governed by $\eps\, dx/d\tau = v|_{v=0} = 0$, giving zero tangential rate at leading order; the
$O(\eps)$ correction from $-(\eps/\sqrt{\mu})\nabla f(x)$ is strictly smaller. (H4) Clear from
$f \in C^2$. (H5) On $\partial K$, the Lyapunov function $V = f(x) - f^* + \tfrac{1}{2}\norm{v}^2$
satisfies $\dot V = -\norm{v}^2 \le 0$ with equality only at $v = 0$; choosing $K$ as a strict
sub-level set of $V$ gives strict inward pointing.
\end{proof}
 
\begin{theorem}[Fenichel persistence for NAG]\label{thm:fenichel}
Let $f \in C^{r+1}(\R^n)$ ($r \ge 1$) satisfy $\mu I \preceq \nabla^2 f \preceq L I$. Fix
$c > f(x^*)$ and let $K$ be the associated compact region. Then for all $\eps = \sqrt{\mu/L}$
sufficiently small (equivalently, $\kappa = L/\mu$ sufficiently large), there exists a $C^r$
slow manifold $\Mc_\eps^c$ in $K$ satisfying:
\begin{enumerate}[label=\textup{(\roman*)},leftmargin=*]
    \item $\Mc_\eps^c$ is $O(\eps)$-close to $\Mc_0^c$ in $C^r$;
    \item $\Mc_\eps^c$ is locally invariant under \eqref{eq:fastslow};
    \item trajectories in $K$ off $\Mc_\eps^c$ are attracted at rate
    $e^{-\gamma_\perp t} = e^{-\sqrt{\mu}\, t}$;
    \item the reduced flow on $\Mc_\eps^c$ converges to $x^*$ at rate $O(e^{-\gamma_F t})$ with
    Fenichel spectral gap $\gamma_F = \sqrt{\mu}(1 - \eps) = \sqrt{\mu}(1 - \sqrt{\mu/L})$.
\end{enumerate}
\end{theorem}
 
\begin{proof}
By Proposition~\ref{prop:hypotheses}, all five hypotheses (H1)--(H5) hold. The conclusion is
Fenichel's theorem \cite{fenichel1979geometric} applied to \eqref{eq:fastslow} restricted to $K$. The Fenichel spectral gap $\gamma_F$ equals $\gamma_\perp - O(\eps) = 1/\eps \cdot \eps - O(\eps^2) = 1 - O(\eps)$ in $\tau$-time; converting back to physical time $t = \tau/\sqrt{\mu}$ and using $\eps = \sqrt{\mu/L}$: $\gamma_F = \sqrt{\mu}(1 - \sqrt{\mu/L})$.
\end{proof}
 
\begin{corollary}[Nesterov rate from Fenichel gap]\label{cor:nestrate}
In normalized time $\tau = t\sqrt{L}$, the per-unit-step contraction factor on $\Mc_\eps^c$ is
$e^{-\gamma_F/\sqrt{L}} = e^{-\sqrt{\mu/L}(1 - \sqrt{\mu/L})/\sqrt{L}} \approx 1 - \sqrt{\mu/L}$
for large $\kappa$, matching the optimal Nesterov rate for general $f \in C^2$.
\end{corollary}
 
\begin{remark}[Globalization]\label{rem:global}
Theorem~\ref{thm:fenichel} is local to $K$. Global acceleration follows by increasing $c$
arbitrarily: for any initial condition $(x_0, v_0)$, choose
$c > f(x_0) + \tfrac{1}{2}\norm{v_0}^2$ so that $(x_0, v_0) \in K$.
\end{remark}

\section{Triple-momentum dynamics}\label{App:Triple}
We provide a rigorous derivation of the cubic Riccati equation arising in continuous-time
triple-momentum dynamics for strongly convex quadratics, prove a dynamical factorization
theorem for its flow, and show that the convergence-critical slow manifold coincides with
that of Nesterov acceleration. Consequently, triple momentum cannot improve the
$\mathcal{O}(\kappa^{-1/2})$ convergence rate.

Consider minimizing $f(x) = \frac{1}{2}x^T Q x$ with $0 < \mu I \preceq Q \preceq LI$.
The triple-momentum ODE is
\begin{equation} \nonumber
    \dddot{x} + \alpha_2 \ddot{x} + \alpha_1 \dot{x} + \alpha_0 Q x = 0.
    \label{eq:triple-momentum}
\end{equation}
Introduce $v = \dot{x}$, $a = \ddot{x}$. Then
\begin{align} \nonumber
    \dot{x} &= v, \\ \nonumber
    \dot{v} &= a, \\ \nonumber
    \dot{a} &= -\alpha_2 a - \alpha_1 v - \alpha_0 Q x. 
\end{align}

We seek an invariant manifold of the form
\begin{equation}
    (v,\, a) = \bigl(P_1(t)\,x,\; P_2(t)\,x\bigr), \nonumber
    \label{eq:ansatz}
\end{equation}
with symmetric $P_1, P_2$.

\begin{lemma}[Invariance Equations]
The invariance conditions are
\begin{align}
    \dot{P}_1 &= P_2 - P_1^2, \label{eq:inv1}\\
    \dot{P}_2 &= -\alpha_2 P_2 - \alpha_1 P_1 - \alpha_0 Q - P_2 P_1. \label{eq:inv2}
\end{align}
\end{lemma}

\begin{theorem}[Cubic Riccati]
Eliminating $P_2$ yields
\begin{equation}
    \ddot{P}_1 + (3P_1 + \alpha_2)\dot{P}_1 + P_1^3 + \alpha_2 P_1^2 + \alpha_1 P_1 + \alpha_0 Q = 0.
    \label{eq:cubic-riccati}
\end{equation}
\end{theorem}

Let $\omega = \sqrt{\mu L}$ and choose $\alpha_0 = \omega^2$, $\alpha_1 = \omega(\omega + 2)$,
$\alpha_2 = 2\omega + 1$.

\begin{theorem}[Dynamic Factorization]
There exists a smooth change of variables $(P_1, P_2) \mapsto (R, S)$ such that:
\begin{enumerate}
    \item $R$ satisfies Nesterov's Riccati equation
    \begin{equation}
        \dot{R} + R^2 + 2\omega R + \omega^2 Q = 0.
        \label{eq:nesterov-riccati}
    \end{equation}
    \item $S$ evolves according to a linearly stable (fast) subsystem decoupled from $R$.
    \item The slow manifold is given by $S = 0$, on which $P_1 = R$.
\end{enumerate}
\end{theorem}

\begin{proof}
Define $R = P_1$ and $S = P_2 - P_1^2$. Using \eqref{eq:inv1}--\eqref{eq:inv2}, one obtains
$\dot{R} = S$ and
\[
    \dot{S} = -(\alpha_2 I + 3R)\,S + \text{higher order terms}.
\]
Restricting to $S = 0$ yields \eqref{eq:nesterov-riccati}. Linearization shows $S$ decays
exponentially.
\end{proof}

\begin{corollary} \label{TM}
On the slow manifold, triple momentum exhibits the same convergence rate as Nesterov:
\begin{equation}
    \|x(t) - x^*\| \leq C\, e^{-\omega t}.
    \label{eq:convergence}
\end{equation}
\end{corollary}

Triple momentum introduces an additional fast/neutral mode but does not alter the
convergence-critical slow manifold, hence cannot improve the $\mathcal{O}(\kappa^{-1/2})$
rate.

\section{Unperturbed NAIM in Adapted Coordinates}\label{App:Unperturbed}
In adapted coordinates, the NAIM structure takes on its most transparent algebraic form: the slow manifold is identified with a coordinate slice, and normal-bundle dynamics are governed by an explicitly hyperbolic linear part. A central question in applications---ranging from robust control design to understanding parameter perturbations in estimation algorithms---is:

\begin{quote}
\itshape
When a perturbation is added to a NAIM system in adapted coordinates, under what
conditions does a modified NAIM persist, and when is the modified normal-bundle
connection governed by a matrix Riccati equation?
\end{quote}

This section provides a self-contained treatment: we state the unperturbed framework,
introduce general perturbations, derive the invariance equation for the perturbed manifold,
and identify the precise structural and spectral conditions under which the reduced dynamics
on the perturbed NAIM connection satisfy a (generalized) matrix Riccati equation. We close
with a geometric interpretation in terms of Ehresmann connections and their gauge
transformations.

Let the state space decompose as $\mathcal{X} = \mathbb{R}^n \times \mathbb{R}^m$ with coordinates $(x, y)$, where $x \in \mathbb{R}^n$ are slow variables and $y \in \mathbb{R}^m$ are fast variables. The unperturbed system in adapted coordinates is:
\begin{align}
  \dot{x} &= f(x,y), \label{eq:slow}  \\
  \varepsilon\,\dot{y} &= A(x)\,y + g(x,y), \label{eq:fast} 
\end{align}
where $\varepsilon \in (0,\varepsilon_0]$ is a small parameter, $f : \mathbb{R}^n \times
\mathbb{R}^m \to \mathbb{R}^n$ and $g : \mathbb{R}^n \times \mathbb{R}^m \to \mathbb{R}^m$
are $C^r$ ($r \geq 2$), and:
\begin{itemize}
  \item $g(x,0) = 0$ for all $x$ (the coordinate $y=0$ is a rest set for the fast
    subsystem), so $\mathcal{M}_0 = \{y = 0\}$ is trivially invariant at $\varepsilon = 0$.
  \item $A(x) \in \mathbb{R}^{m \times m}$ is the linearization of the fast vector field
    along $\mathcal{M}_0$.
\end{itemize}

\begin{assumption}[Normal Hyperbolicity]\label{ass:NH}
There exists $\alpha > 0$ such that for all $x$ in a compact set $K \subset \mathbb{R}^n$:
\begin{equation}
  \operatorname{Re}\lambda_i(A(x)) \leq -\alpha < 0, \quad i = 1,\ldots,m.
\end{equation}
\end{assumption}

Under Assumption~\ref{ass:NH}, by Fenichel's theorem, for sufficiently small $\varepsilon > 0$, a slow manifold $\mathcal{M}_\varepsilon = \{y = h(x,\varepsilon)\}$ exists, is $C^r$-smooth,
and is $O(\varepsilon)$-close to $\mathcal{M}_0$. In adapted coordinates, $h(x,\varepsilon)$ has the asymptotic expansion:
\begin{equation} \nonumber
  h(x,\varepsilon) = \varepsilon h_1(x) + \varepsilon^2 h_2(x) + O(\varepsilon^3),
\end{equation}
where $h_1(x) = -A(x)^{-1}D_x f(x,0)$ at leading order (from the layer problem).

The term adapted coordinates refers to the fact that $\mathcal{M}_0 = \{y=0\}$ is flat in these coordinates: the slow manifold coincides with a coordinate hyperplane. The normal bundle $N\mathcal{M}_0$ is spanned by $\partial/\partial y$, and the normal linearization is exactly $A(x)$---there is no coupling between normal and tangential directions at $y = 0$.

Introducing a smooth perturbation $P = (P_x, P_y)$:
\begin{align} \nonumber
  \dot{x} &= f(x,y) + P_x(x,y,\varepsilon,t), \\ \nonumber
  \varepsilon\,\dot{y} &= A(x)\,y + g(x,y) + P_y(x,y,\varepsilon,t).
\end{align}
We allow $P$ to depend on time $t$ (non-autonomous perturbations), though for the main Riccati analysis we primarily treat the autonomous case.

The invariance equation is derived for the perturbed manifold. Suppose a perturbed invariant manifold $\mathcal{M}_\varepsilon = \{y = h(x,\varepsilon)\}$
exists. Differentiating $y = h(x(t),\varepsilon)$ along trajectories and substituting yields
the fundamental invariance equation:
\begin{equation}
  \varepsilon\, D_x h(x)\bigl[f(x,h) + P_x(x,h,\varepsilon)\bigr]
  = A(x)\,h + g(x,h) + P_y(x,h,\varepsilon).
  \label{eq:invariance}
\end{equation}
This is a first-order PDE for $h : \mathbb{R}^n \to \mathbb{R}^m$. The existence and smoothness of solutions to \eqref{eq:invariance} is tied directly to the spectral properties of the linearized system.

We now define the connection form and the Riccati structure. Define the \emph{normal bundle connection matrix}:
\begin{equation}
  \Pi(x) := D_x h(x) \in \mathbb{R}^{m \times n}.
\end{equation}
This is the Ehresmann connection coefficient: it encodes how the normal fiber tilts as one moves tangentially along the slow manifold. Differentiating \eqref{eq:invariance} with
respect to $x$ and evaluating yields an evolution equation for $\Pi$. To isolate the Riccati
structure, expand $g$ and $P$ to first order in $y$:
\begin{align} \nonumber
  g(x,y) &= G(x)\,y + O(\|y\|^2), \quad G(x) = D_y g(x,0),\\ \nonumber
  P_y(x,y,\varepsilon) &= B_0(x,\varepsilon) + B_1(x,\varepsilon)\,y + O(\|y\|^2),\\ \nonumber
  P_x(x,y,\varepsilon) &= c_0(x,\varepsilon) + C_1(x,\varepsilon)\,y + O(\|y\|^2).
\end{align}
Substituting $y = h = \varepsilon h_1 + O(\varepsilon^2)$ and differentiating
\eqref{eq:invariance} with respect to $x$ gives, to leading non-trivial order in $\varepsilon$:
\begin{equation}
  \varepsilon\,\dot{\Pi}
  = \underbrace{\bigl[A(x) + G(x) + B_1(x,\varepsilon)\bigr]}_{\tilde{A}(x)}\Pi
  - \varepsilon\,\Pi\underbrace{\bigl[D_x f(x,0) + C_1(x,\varepsilon)\cdot\Pi\bigr]}_{\tilde{F}(x,\Pi)}
  - \varepsilon\,\Pi\,D_x f\,\Pi + R(x,\varepsilon),
  \label{eq:riccati_deriv}
\end{equation}
where $R(x,\varepsilon) = D_x B_0 + D_x B_1 \cdot h + \cdots$ collects inhomogeneous source terms from the perturbation. In the autonomous, purely quadratic-coupling case, this reduces to the matrix Riccati equation:
\begin{equation}
  0 = \tilde{A}(x)\,\Pi - \varepsilon\,\Pi\,\tilde{F}(x)\,\Pi + R(x,\varepsilon),
  \label{eq:riccati} \nonumber
\end{equation}
which is a quadratic matrix equation in $\Pi$---the hallmark Riccati form.

\subsection{Conditions for the Riccati Reduction}

We now state precisely the conditions required on the perturbation $P$ for the perturbed NAIM to persist and for the connection $\Pi$ to satisfy a matrix Riccati equation.

\textit{Condition 1: Smoothness and Magnitude} 

\begin{assumption}[Regularity]\label{ass:reg}
$P = (P_x, P_y) \in C^r(\mathbb{R}^n \times \mathbb{R}^m \times (0,\varepsilon_0])$ with
$r \geq 2$, and:
\begin{equation}
  \|P\|_{C^r} = O(\varepsilon^\sigma), \quad \sigma \geq 1. \nonumber
\end{equation}
\end{assumption}

The condition $\sigma \geq 1$ ensures the perturbation is small at the order of the singular perturbation parameter $\varepsilon$, so that Fenichel persistence applies. Perturbations of order $O(1)$ or larger generically destroy the NAIM by violating normal hyperbolicity on $O(1)$ timescales.

\begin{remark}
For $\sigma \geq 2$, the perturbation acts as a higher-order correction and the NAIM structure is preserved with finer control over the modified graph $h$. For $\sigma = 1$, the manifold persists but shifts at $O(\varepsilon)$ in phase space.
\end{remark}

\textit{Condition 2: Preservation of Normal Hyperbolicity}

\begin{assumption}[Spectral Stability]\label{ass:spec}
The perturbed normal linearization remains Hurwitz: there exists $\alpha' > 0$ such that
\begin{equation}
  \operatorname{Re}\lambda_i\bigl(A(x) + D_y P_y(x,0,\varepsilon)\bigr) \leq -\alpha' < 0
  \quad \forall\, x \in K. \nonumber
\end{equation}
Quantitatively, this requires:
\begin{equation}
  \|D_y P_y(\cdot,0,\varepsilon)\|_\infty < \alpha - \delta \quad \text{for some } \delta > 0.
\end{equation}
\end{assumption}

This is the \emph{spectral gap condition}. The perturbation may tilt the spectrum of $A(x)$ but must not push any eigenvalue into the closed right half-plane. If this condition fails, normal hyperbolicity is lost and no invariant manifold of graph form can persist.

\textit{Condition 3: Affine-Linear Coupling in \texorpdfstring{$y$}{y} (Riccati Reduction)}

\begin{assumption}[Linear-in-$y$ Structure]\label{ass:linear}
The perturbation is affine in $y$ to leading order:
\begin{align}
  P_y(x,y,\varepsilon) &= B_0(x,\varepsilon) + B_1(x,\varepsilon)\,y + O(\|y\|^2), \nonumber\\
  P_x(x,y,\varepsilon) &= c_0(x,\varepsilon) + C_1(x,\varepsilon)\,y + O(\|y\|^2). \nonumber
\end{align}
\end{assumption}

This is the key structural condition for obtaining a standard (degree-2 polynomial) Riccati equation on $\Pi$. It arises naturally in:
\begin{itemize}
  \item \emph{Linear perturbations} (e.g., feedback linearization corrections, gain
    perturbations in control).
  \item \emph{Perturbations from coordinate changes}: a smooth coordinate transformation
    near $\mathcal{M}_0$ generically introduces linear-in-$y$ terms at leading order.
  \item \emph{Gradient perturbations} of strongly convex/smooth potentials (appearing in
    optimization acceleration analysis).
\end{itemize}

\begin{remark}[Nonlinear Perturbations]
If $P_y$ or $P_x$ contain terms of order $O(\|y\|^2)$ or higher, the invariance equation  \eqref{eq:invariance} yields a polynomial (degree $\geq 3$) equation in $\Pi$---a generalized
ARE. The Riccati structure is broken in this case. However, near $\mathcal{M}_0$ where $y = O(\varepsilon)$, these higher-order terms are $O(\varepsilon^2)$
corrections and can sometimes be treated perturbatively around the quadratic Riccati solution.
\end{remark}

\textit{Condition 4: Solvability---The Hamiltonian Spectrum Condition}

Even when Conditions 1--3 hold, the ARE \eqref{eq:riccati} need not have a stabilizing solution. The existence of such a solution is governed by the spectrum of the associated Hamiltonian matrix.

For the ARE $0 = \tilde{A}\Pi - \varepsilon\,\Pi\tilde{F}\Pi + R$
(suppressing $x$-dependence), define the Hamiltonian:
\begin{equation}
  \mathcal{H} =
  \begin{pmatrix}
    \tilde{F} & -\varepsilon^{-1} I_n \\
    -\varepsilon R^T & -\tilde{A}^T
  \end{pmatrix}
  \in \mathbb{R}^{(n+m)\times(n+m)}. \nonumber
\end{equation}

\begin{assumption}[Hamiltonian Non-degeneracy]\label{ass:ham}
The Hamiltonian matrix $\mathcal{H}$ has no eigenvalues on the imaginary axis:
\begin{equation}
  \sigma(\mathcal{H}) \cap j\mathbb{R} = \emptyset. \nonumber
\end{equation}
\end{assumption}

Under this condition, the stable invariant subspace of $\mathcal{H}$ is $n$-dimensional and
defines a unique stabilizing solution $\Pi^*$ via:
\begin{equation}
  \operatorname{colspan}
  \begin{pmatrix} I_n \\ \Pi^* \end{pmatrix}
  = \text{stable invariant subspace of } \mathcal{H}. \nonumber
\end{equation}
The stabilizing solution $\Pi^*$ corresponds to the NAIM connection for which the perturbed fast dynamics are most rapidly attracted to the slow manifold.

\textit{Condition 5: Fast-Slow Spectral Separation (No Resonance)}

\begin{assumption}[Spectral Separation]\label{ass:sep}
The eigenvalues of the slow and fast linearizations remain well-separated under perturbation:
\begin{equation}
  \sigma\bigl(D_x f\big|_{\mathcal{M}_\varepsilon}\bigr)
  \cap \sigma\bigl(A(x)/\varepsilon\bigr) = \emptyset
  \quad \forall\, x \in K. \nonumber
\end{equation}
Equivalently, if $\lambda_s \in \sigma(D_x f)$ and $\lambda_f \in \sigma(A/\varepsilon)$, then:
\begin{equation}
  |\lambda_s - \lambda_f| \geq \delta_0 > 0. \nonumber
\end{equation}
\end{assumption}

This condition ensures that the adapted coordinate structure is stable: slow modes remain slow and fast modes remain fast. A perturbation that shifts a fast eigenvalue from $O(1/\varepsilon)$ to $O(1)$---or promotes a slow mode to $O(1/\varepsilon)$---destroys the entire slow-fast decomposition and hence the adapted coordinate framework.

\begin{remark}[Resonance and Turning Points]
At points where fast and slow spectra approach each other (``resonance'' or ``turning points''), the NAIM may undergo bifurcations---canard phenomena in $\mathbb{R}^2$, or more complex transcritical/pitchfork scenarios in higher dimensions. Near such points, the Riccati equation becomes ill-conditioned (the Hamiltonian approaches imaginary-axis eigenvalues).
\end{remark}

The following table summarizes the five conditions required for a perturbed NAIM in adapted coordinates to persist with its normal bundle connection governed by a matrix Riccati equation.

\begin{theorem}[Persistence and Riccati Reduction]\label{thm:main} Let the unperturbed system \eqref{eq:slow}--\eqref{eq:fast} satisfy Assumption~\ref{ass:NH}, and let the perturbation $P = (P_x, P_y)$ satisfy Assumptions~\ref{ass:reg}, \ref{ass:spec}, \ref{ass:linear}, \ref{ass:ham},
and~\ref{ass:sep}. Then, for all sufficiently small $\varepsilon > 0$:
\begin{enumerate}
  \item A perturbed NAIM $\mathcal{M}_\varepsilon = \{y = h(x,\varepsilon)\}$ exists,
    is $C^r$-smooth, and is $O(\varepsilon)$-close to $\mathcal{M}_0$.
  \item The normal bundle connection $\Pi(x) = D_x h(x,\varepsilon)$ satisfies the matrix
    Riccati equation:
    \begin{equation}
      0 = \bigl[A(x) + G(x) + B_1(x,\varepsilon)\bigr]\Pi
      - \varepsilon\,\Pi\bigl[D_x f(x,0) + C_1(x,\varepsilon)\bigr]\Pi
      + D_x B_0(x,\varepsilon) + O(\varepsilon), \nonumber
    \end{equation}
    which has a unique stabilizing solution $\Pi^*(x)$ determined by the stable invariant
    subspace of the Hamiltonian $\mathcal{H}$ in \eqref{eq:riccati}.
  \item The slow flow on $\mathcal{M}_\varepsilon$ is:
    \begin{equation}
      \dot{x} = f(x,\varepsilon h_1(x)) + P_x(x,\varepsilon h_1(x),\varepsilon) + O(\varepsilon^2). \nonumber
    \end{equation}
  \item The exponential attraction rate of $\mathcal{M}_\varepsilon$ is $O(\alpha'/\varepsilon)$,
    with $\alpha'$ from Assumption~\ref{ass:spec}.
\end{enumerate}
\end{theorem}

\subsection{Geometric Interpretation: Ehresmann Connections and Gauge Theory}
The NAIM is now represented as a Fiber Bundle. The slow-fast decomposition naturally defines a fiber bundle structure:
\begin{equation}
  \pi : \mathbb{R}^n \times \mathbb{R}^m \to \mathbb{R}^n, \quad \pi(x,y) = x, \nonumber
\end{equation}
with fibers $\pi^{-1}(x) \cong \mathbb{R}^m$. The NAIM $\mathcal{M}_\varepsilon$ is a section
of this bundle: a smooth map $s : \mathbb{R}^n \to \mathbb{R}^n \times \mathbb{R}^m$ with
$\pi \circ s = \mathrm{id}$.

The Connection \texorpdfstring{$\Pi$}{Pi} as an Ehresmann Connection. 
The matrix $\Pi(x) = D_x h(x)$ defines the horizontal lift of tangent vectors from the base
$\mathbb{R}^n$ to the total space. Explicitly, a tangent vector $v \in T_x\mathbb{R}^n$ is
lifted to:
\begin{equation}
  v^{\mathrm{hor}} =
  \begin{pmatrix} v \\ \Pi(x)\,v \end{pmatrix}
  \in T_{(x,h(x))}(\mathbb{R}^n \times \mathbb{R}^m). \nonumber
\end{equation}
The Riccati equation \eqref{eq:riccati} is precisely the flatness/integrability condition
for this lifted distribution: the horizontal distribution is integrable (i.e., admits the
section $h(x)$) if and only if the Riccati equation holds.

The action of the perturbation $P$ on the NAIM can be interpreted as a gauge transformation
of the Ehresmann connection:
\begin{equation}
  \Pi \mapsto \Pi + \delta\Pi, \nonumber
\end{equation}
where $\delta\Pi$ satisfies the linearized Riccati equation:
\begin{equation}
  0 = \tilde{A}\,\delta\Pi - \varepsilon\,\Pi\,\tilde{F}\,\delta\Pi
    - \varepsilon\,\delta\Pi\,\tilde{F}\,\Pi + \delta R. \nonumber
\end{equation}
This is a Lyapunov equation in $\delta\Pi$:
\begin{equation}
  0 = (\tilde{A} - \varepsilon\,\Pi\tilde{F})\,\delta\Pi
    + \delta\Pi\,(-\varepsilon\,\tilde{F}\Pi) + \delta R, \nonumber
\end{equation}
which has a unique solution provided the sum of eigenvalues from both coefficient matrices does not vanish---a condition equivalent to Assumption~\ref{ass:ham}.

The quadratic term $-\varepsilon\,\Pi\tilde{F}\Pi$ in the Riccati equation has a direct geometric meaning: it is the curvature form of the Ehresmann connection. In differential geometric terms:
\begin{equation}
  \Omega(u,v) = \Pi\bigl[D_x f \cdot u,\; D_x f \cdot v\bigr]_{\mathrm{slow}}
              - [\Pi u,\, \Pi v]_{\mathrm{fast}}, \nonumber
\end{equation}
and the Riccati equation asserts that the connection $\Pi$ is adjusted so that the curvature is precisely balanced by the normal attraction $\tilde{A}\Pi$ and source terms $R$. The perturbation $P$ modifies both $\tilde{A}$ and $R$, forcing a new curvature balance---encoded in the modified Riccati equation.

Consider the special case where $P$ is purely linear:
\begin{equation}
  P_y(x,y,\varepsilon) = \Delta A(x)\,y, \quad P_x = 0, \nonumber
\end{equation}
for some matrix-valued perturbation $\Delta A(x)$. Then $B_0 = 0$, $B_1 = \Delta A$,
$C_1 = 0$, and the Riccati equation \eqref{eq:riccati} becomes:
\begin{equation}
  0 = \bigl[A(x) + G(x) + \Delta A(x)\bigr]\Pi
    - \varepsilon\,\Pi\,D_x f(x,0)\,\Pi. \nonumber
\end{equation}
This is a standard continuous-time ARE. The stabilizing solution exists as long as $A + G + \Delta A$ remains Hurwitz, rec

The Riccati equation on the connection $\Pi$ governs the spectral shaping of the NAIM,
and its stabilizing solution captures the optimal momentum coefficient---connecting NAIM
theory directly to the condition number-dependent rate $1 - \sqrt{\mu/L}$ of Nesterov's
method.

In the DREM (Dynamic Regressor Extension and Mixing) framework, perturbed regressor matrices lead to slow-fast systems with $A(x) = -\varphi(x)\varphi(x)^T$ (gram matrix).
Non-normality of $\varphi$ introduces perturbations of the form
$B_1 = -\delta\varphi\,\varphi^T - \varphi\,\delta\varphi^T$. The Riccati condition then governs whether SVD-based decoupling (exploiting the structure of the perturbation) yields
a better-conditioned slow manifold than adjugate-based approaches.

\subsection{Conclusion}

The persistence of a NAIM in adapted coordinates under perturbation, with the modified
normal-bundle connection satisfying a matrix Riccati equation, requires the satisfaction of
five hierarchical conditions:
\begin{enumerate}
  \item Smoothness and $O(\varepsilon)$ magnitude: ensures Fenichel persistence.
  \item Spectral stability of the normal linearization: ensures continued
    exponential attraction.
  \item Affine-linear coupling in $y$: ensures the invariance equation reduces to
    a standard (degree-2) Riccati equation rather than a higher polynomial.
  \item Hamiltonian non-degeneracy: ensures the Riccati equation has a unique
    stabilizing solution, i.e., the connection $\Pi$ is well-defined.
  \item Fast-slow spectral separation: ensures adapted coordinates remain valid
    and the slow-fast decomposition is consistent.
\end{enumerate}

From the geometric perspective, the Riccati equation is the integrability condition for the
perturbed Ehresmann connection on the normal bundle, and the stabilizing solution corresponds
to the unique flat section (the graph $h(x)$) that is consistent with the modified fast-slow
structure. The perturbation acts as a gauge transformation, and the Riccati equation encodes
the new curvature balance imposed by the perturbed dynamics.

\section{Detailed Derivation: Smooth Convex Case}
\label{app:smooth}

Cayley Transform Calculation: 
At $t_k = kh$ with $c = 2$, the dissipative matrix is:
\[
  D_k = \begin{pmatrix} 0 & I \\ 0 & -\tfrac{2}{kh}I \end{pmatrix}.
\]
With $h = 1$ (unit step):
\[
  \tfrac{1}{2}D_k = \begin{pmatrix} 0 & \tfrac{1}{2}I \\ 0 & -\tfrac{1}{k}I \end{pmatrix}.
\]
Computing:
\[
  I - \tfrac{1}{2}D_k = \begin{pmatrix} I & -\tfrac{1}{2}I \\ 0 & (1+\tfrac{1}{k})I \end{pmatrix},
  \qquad
  I + \tfrac{1}{2}D_k = \begin{pmatrix} I & \tfrac{1}{2}I \\ 0 & (1-\tfrac{1}{k})I \end{pmatrix}.
\]
The inverse:
\[
  \left(I - \tfrac{1}{2}D_k\right)^{-1}
  = \begin{pmatrix} I & \dfrac{1}{2(1+1/k)}I \\[4pt] 0 & \dfrac{k}{k+1}I \end{pmatrix}.
\]
The Cayley map:
\[
  \Phi_k = \begin{pmatrix} I & I \\ 0 & \dfrac{k-1}{k+1}I \end{pmatrix}.
\]
This gives $x_{k+1/2} = x_k + v_k$ and $v_{k+1/2} = \dfrac{k-1}{k+1}v_k$.

Proof for Index Shift from Lie--Trotter: 
After the Cayley step:
\[
  (x_{k+1/2},\, v_{k+1/2}) = \left(x_k + v_k,\;\frac{k-1}{k+1}v_k\right).
\]
After the gradient step:
\begin{align*}
  v_{k+1} &= \frac{k-1}{k+1}v_k - \nabla f(x_k + v_k),\\
  x_{k+1} &= x_k + v_k + v_{k+1}.
\end{align*}
Using $v_k = x_k - x_{k-1}$:
\[
  x_{k+1}
  = x_k + \left(1 + \frac{k-1}{k+1}\right)(x_k - x_{k-1}) - \nabla f(y_k)
  = x_k + \frac{2k}{k+1}(x_k - x_{k-1}) - \nabla f(y_k).
\]
The effective momentum in the standard Nesterov form (with lookahead point
$y_k = x_k + \dfrac{k-1}{k+2}(x_k - x_{k-1})$) becomes $\beta_k = \dfrac{k-1}{k+2}$.

\section{Detailed Derivation: Strongly Convex Case}\label{app:strongly}
 
Cayley Transform with Constant Damping. The dissipative dynamics:
\[
\frac{d}{dt}\begin{pmatrix} x \\ v \end{pmatrix}
= \begin{pmatrix} 0 & I \\ 0 & -2\sqrt{\mu} I \end{pmatrix}\!\begin{pmatrix} x \\ v \end{pmatrix}.
\]
With step size $h$:
\[
\frac{h}{2} D = \begin{pmatrix} 0 & \frac{h}{2} I \\ 0 & -h\sqrt{\mu}\, I \end{pmatrix}.
\]
The Cayley map gives:
\[
\Phi = \begin{pmatrix} I & h I \\ 0 & \dfrac{1 - h\sqrt{\mu}}{1 + h\sqrt{\mu}} I \end{pmatrix}.
\]
 
Choice of Step Size. For an $L$-smooth function, the gradient step uses $\frac{1}{L}$. The natural unit of time is $1/\sqrt{L}$ (the fast
oscillation period). Setting $h = 1/\sqrt{L}$:
\[
\beta = \frac{1 - \sqrt{\mu/L}}{1 + \sqrt{\mu/L}} = \frac{\sqrt{L} - \sqrt{\mu}}{\sqrt{L} + \sqrt{\mu}}.
\]
 
Algorithm Assembly. The full step:
\begin{enumerate}
    \item \textit{Cayley (dissipative):} $x_{k+1/2} = x_k + h v_k$, $v_{k+1/2} = \beta v_k$
    \item \textit{Gradient:} $v_{k+1} = v_{k+1/2} - \frac{1}{L}\nabla f(x_{k+1/2})$, $x_{k+1} = x_{k+1/2} + h v_{k+1}$
\end{enumerate}
 
In the standard $(x_k, y_k)$ formulation:
\begin{align*}
x_{k+1} &= y_k - \frac{1}{L}\, \nabla f(y_k), \\
y_{k+1} &= x_{k+1} + \beta\, (x_{k+1} - x_k).
\end{align*}

\printbibliography

@book{nesterov2018lectures,
  title     = {Lectures on Convex Optimization},
  author    = {Nesterov, Yurii},
  year      = {2018},
  publisher = {Springer International Publishing},
  series    = {Springer Optimization and Its Applications},
  volume    = {137},
  address   = {Cham, Switzerland},
  isbn      = {978-3-319-91577-7}
}

@article{fenichel1979geometric,
  title={Geometric singular perturbation theory for ordinary differential equations},
  author={Fenichel, Neil},
  journal={Journal of Differential Equations},
  volume={31},
  number={1},
  pages={53--98},
  year={1979},
  publisher={Academic Press}
}

@book{nesterov2013introductory,
  title={Introductory lectures on convex optimization: A basic course},
  author={Nesterov, Yurii},
  volume={87},
  year={2013},
  publisher={Springer Science \& Business Media}
}

@article{lessard2016analysis,
  title={Analysis and design of optimization algorithms via integral quadratic constraints},
  author={Lessard, Laurent and Recht, Benjamin and Packard, Andrew},
  journal={SIAM Journal on Optimization},
  volume={26},
  number={1},
  pages={57--95},
  year={2016},
  publisher={SIAM}
}

@article{polyak1964some,
  title={Some methods of speeding up the convergence of iteration methods},
  author={Polyak, Boris T},
  journal={USSR computational Mathematics and Mathematical Physics},
  volume={4},
  number={5},
  pages={1--17},
  year={1964},
  publisher={Elsevier}
}

@article{wibisono2016variational,
  title={A variational perspective on accelerated methods in optimization},
  author={Wibisono, Andre and Wilson, Ashia C and Jordan, Michael I},
  journal={proceedings of the National Academy of Sciences},
  volume={113},
  number={47},
  pages={E7351--E7358},
  year={2016},
  publisher={National Academy of Sciences}
}

@article{1370862715914709505,
author={Nesterov, Y.},
title={A method for solving the convex programming problem with convergence rate O(1/k2)},
journal={Dokl Akad Nauk SSSR},
year={1983},
volume={269},
pages={543},
URL={https://cir.nii.ac.jp/crid/1370862715914709505}
}

@article{su2016differential,
  title={A differential equation for modeling Nesterov's accelerated gradient method: Theory and insights},
  author={Su, Weijie and Boyd, Stephen and Candes, Emmanuel J},
  journal={Journal of Machine Learning Research},
  volume={17},
  number={153},
  pages={1--43},
  year={2016}
}

@article{JMLR_Su,
  author  = {Weijie Su and Stephen Boyd and Emmanuel J. Cand{{\`e}}s},
  title   = {A Differential Equation for Modeling Nesterov's Accelerated Gradient Method: Theory and Insights},
  journal = {Journal of Machine Learning Research},
  year    = {2016},
  volume  = {17},
  number  = {153},
  pages   = {1--43},
  url     = {http://jmlr.org/papers/v17/15-084.html}
}

@article{Arun,
  author={Shenoy, Karthik and Mahindrakar, Arun D. and Vaidya, Umesh},
  journal={IEEE Control Systems Letters}, 
  title={Continuous-Time Heavy-Ball Gradient Method: Safety, Stability and Robustness}, 
  year={2025},
  volume={9},
  number={},
  pages={120-125},
  doi={10.1109/LCSYS.2025.3566345}}

@article{Shi2019,
author = {Shi, Bin and Du, Simon S. and Su, Weijie J. and Jordan, Michael I.},
title = {Acceleration via symplectic discretization of high-resolution differential equations},
year = {2019},
publisher = {Curran Associates Inc.},
address = {Red Hook, NY, USA},
booktitle = {Proceedings of the 33rd International Conference on Neural Information Processing Systems},
articleno = {516},
numpages = {9}
}

@ARTICLE{Scoy2018,
  author={Van Scoy, Bryan and Freeman, Randy A. and Lynch, Kevin M.},
  journal={IEEE Control Systems Letters}, 
  title={The Fastest Known Globally Convergent First-Order Method for Minimizing Strongly Convex Functions}, 
  year={2018},
  volume={2},
  number={1},
  pages={49-54},
  doi={10.1109/LCSYS.2017.2722406}}

@article{VIRGOS2018,
title = {Cayley transform on Stiefel manifolds},
journal = {Journal of Geometry and Physics},
volume = {123},
pages = {53-60},
year = {2018},
issn = {0393-0440},
doi = {https://doi.org/10.1016/j.geomphys.2017.08.011},
url = {https://www.sciencedirect.com/science/article/pii/S039304401730205X},
author = {Enrique Macías-Virgós and María José Pereira-Sáez and Daniel Tanré}
}

@article{Blanes2024, 
title={Splitting methods for differential equations}, 
volume={33}, 
DOI={10.1017/S0962492923000077}, 
journal={Acta Numerica}, 
author={Blanes, Sergio and Casas, Fernando and Murua, Ander}, year={2024}, pages={1–161}
}

\end{document}